\documentclass[11pt]{article}

\usepackage[T1]{fontenc}
\usepackage{lmodern}
\usepackage{amsmath,amssymb,amsthm,mathtools,mathrsfs}
\usepackage{microtype}
\usepackage{graphicx}
\usepackage[a4paper,margin=30mm]{geometry}
\usepackage[hidelinks]{hyperref}

\newtheorem{theorem}{Theorem}[section]
\newtheorem{proposition}[theorem]{Proposition}
\newtheorem{lemma}[theorem]{Lemma}
\newtheorem{conjecture}[theorem]{Conjecture}
\theoremstyle{remark}
\newtheorem{remark}[theorem]{Remark}

\newcommand{\Fock}{\mathcal{F}}
\newcommand{\Dom}{\mathcal{D}}
\newcommand{\C}{\mathbb{C}}
\newcommand{\Nzero}{\mathbb{N}_{0}}
\newcommand{\id}{\mathbf{1}}

\title{Operator-theoretic and analytic properties of a driven
single-mode Kerr cavity in Bargmann space}
\author{Maciej Janowicz}
\date{}

\begin{document}

\maketitle

\begin{abstract}
The stationary spectral problem for a coherently driven single-mode
electromagnetic field in a Kerr medium is studied.  For arbitrary detuning
and driving amplitude and positive Kerr coupling, the Hamiltonian is shown
to be self-adjoint, essentially self-adjoint on its finite-particle core,
bounded below, and to have compact resolvent; consequently the spectrum is
purely discrete.  The eigenvalue equation is formulated in the
Bargmann--Fock representation as a second-order differential equation for
an entire function.  A minimal-solution continued fraction gives, for non-zero drive,
an exact scalar spectral condition $\Xi(E)=0$, with every eigenvalue simple.
A Liouville transformation removes the equation's first derivative, places
it in normal form, classifies its singularities at the origin and at
infinity, and identifies its precise specialization within the
double-confluent Heun family.  An Olver-type Volterra construction then
gives uniquely normalized sectorial solutions at infinity with explicit
error bounds.  In the strong-drive limit, the leading turning-point energy
coefficient is proved rigorously, while its next-order coefficient is
conditional on an explicitly stated global-connection conjecture; an
independently derived formal Bogoliubov expansion, and a second,
complex-WKB treatment resting on a separately conjectured orientation
hypothesis, reproduce the same coefficients only as mutual consistency
checks, not as a proof of either conjecture.  A complementary formal
weak-drive perturbation expansion is independently cross-checked against
the Bargmann coefficient recurrence, and the upper edge of the resulting
tridiagonal series is summed exactly in closed form, in analogy with the
diagonal resummation of Bender and Bettencourt.  Exploratory WKB iteration
portraits, presented separately, are visual aids only and establish
nothing about the global connection, which remains conjectural.
\end{abstract}

\section{Introduction}
\label{sec:introduction}

A single bosonic mode subject to a Kerr nonlinearity and coherent
driving is considered.  In quantum optics this operator is the rotating-frame Hamiltonian of
a one-photon-driven Kerr resonator; it also arises from a near-resonantly
driven Duffing oscillator after a rotating-wave approximation.  Its
eigenvalues then represent rotating-frame energies, or representatives of
laboratory-frame Floquet quasienergies.  The resulting multiphoton resonances,
avoided quasienergy crossings, and antiresonant response have been studied
perturbatively and numerically
\cite{DykmanFistul2005Antiresonance,PeanoThorwart2006Duffing,LeytonPeanoThorwart2012Noise}.

The driven Kerr Hamiltonian considered here is, up to the sign convention for
the detuning, precisely the single-site case of the driven Bose--Hubbard
Hamiltonian studied for arrays of coupled nonlinear cavities, obtained by
dropping the inter-site hopping term \cite{LeBoiteOrsoCiuti2013DrivenBH}.  The
same rotating-frame Hamiltonian, in the same convention as
\cite{LeBoiteOrsoCiuti2013DrivenBH}, is written out explicitly in
\cite{VicentiniEtAl2018CriticalSlowing}.

Particularly close to the present model, Maslova et al.\ studied the same
one-photon Kerr Hamiltonian (with a real drive), obtained its
eigenstates by Hamiltonian-matrix diagonalization, and displayed selected
states in the complex coherent-amplitude plane
\cite{MaslovaEtAl2019Squeezing}.  Such coherent-state and phase-space
descriptions are established prior art.  They should also be distinguished
from the extensive literature on driven-dissipative stationary states, where
exact generalized-\(P\) and Segal--Bargmann constructions solve a master
equation rather than the closed eigenvalue problem
\cite{DrummondWalls1980Bistability,RobertsClerk2020Kerr}.

The same physical setting also has experimentally established nonlinear-resonator
realizations: driven Josephson resonators exhibit dynamical switching in the
Duffing regime \cite{SiddiqiEtAl2004Bifurcation}, while strong single-photon
Kerr evolution has been resolved through Husimi and Wigner tomography
\cite{KirchmairEtAl2013Kerr}.  Extensions of the rotating-frame model by
higher-order nonlinearities have likewise been studied numerically in connection
with multiphoton resonances \cite{AnikinEtAl2021Multiphoton}; they are useful
physical context but are not the pure closed Kerr eigenproblem considered here.

Our aim is a description of the holomorphic Bargmann--Fock eigenfunctions,
including their growth and pointwise structure throughout the complex plane, together with
an analytic or quasi-analytic spectral condition not based on direct
Hamiltonian-matrix diagonalization.  The exact continued-fraction condition
below supplies the spectral starting point, while the sectorial asymptotic
constructions expose the remaining global connection problem.
Independently checked formal drive expansions and exploratory portraits
complete the hierarchy of results.

The main body of the paper is organized as follows.  In
Section~\ref{sec:hamiltonian-bargmann}, we introduce the driven Kerr
Hamiltonian, derive it from the laboratory-frame model by passing to the
frame rotating at the drive frequency, and cast the eigenvalue problem in
the Bargmann representation as a second-order differential equation.
Section~\ref{sec:operator-properties} comprises the self-adjointness,
lower-semiboundedness, and compact-resolvent results that together
establish discreteness of the spectrum.
Section~\ref{sec:heunwi} identifies the Bargmann eigenvalue equation with
the Whittaker--Ince equation, gives its normalized formal solution
explicitly by a three-term recurrence, and characterizes the spectrum by
the resulting continued-fraction equation.
Section~\ref{sec:liouville} removes the first-order derivative from the
Bargmann equation by a local Liouville transformation, casts the resulting
equation in normal form, classifies its singularities at the origin and at
infinity together with their formal local behaviours, and identifies the
normal form as a degenerate parameter case of the double-confluent Heun
family, including its Leaver--Figueiredo and Whittaker--Ince
specializations.
Section~\ref{sec:sectorial-infinity} constructs formal and then rigorous
asymptotic solutions of the equation at the irregular singular point at
infinity, comparing a direct formal series expansion with the ordinary
Liouville--Green approximation before proving, following Olver's method,
the existence of uniquely normalized sectorial solutions with explicit
error bounds.
Section~\ref{sec:strong-drive} analyzes the strong-drive limit, in which
two turning points coalesce and a parabolic-cylinder comparison equation
governs the local connection problem, states the resulting global
quantization condition as a conjecture, and independently checks the
conjectured energy asymptotics by a Bogoliubov operator expansion.
Section~\ref{sec:sectorial-solutions} develops an independent complex-WKB
treatment of the strong-drive equation directly in the Fedoryuk normal
form, identifies two candidate sectorial energy branches, and states, as a
further hypothesis, the global connection conditions under which one of
them reproduces the conjectured spectrum.
Section~\ref{sec:weak-drive} develops a regular perturbation expansion in
the weak-drive parameter \(\lambda=|F|/V\), obtained after a
number-operator phase rotation, and shows that the upper edge of the
resulting tridiagonal series sums exactly, in analogy with the diagonal
resummation of Bender and Bettencourt.
Section~\ref{sec:summary} collects, as a numbered list with explicit
epistemic status, every proven, conditional, and conjectural claim made in
the paper, and Section~\ref{sec:outlook} closes with a brief synthetic
assessment of what the work as a whole establishes.

\section{Hamiltonian and Bargmann representation}
\label{sec:hamiltonian-bargmann}

Let \(a\) and \(a^\dagger\) be the annihilation and creation operators on the
one-mode bosonic Fock space \(\Fock=\ell^2(\Nzero)\), with number operator
\(N=a^\dagger a\).  We study
\begin{equation}
  H=\frac{V}{2}a^{\dagger 2}a^2+\hbar\Delta a^\dagger a
    +F a^\dagger+\overline{F}\,a,
  \label{eq:H}
\end{equation}
where
\begin{equation}
  V>0,\qquad \Delta\in\mathbb{R},\qquad F\in\mathbb{C}.
  \label{eq:parameters}
\end{equation}
For physical orientation, consider the laboratory-frame Kerr-mode model
\begin{equation}
 H_{\mathrm{lab}}(t)=\hbar\omega_cN
 +\frac{V}{2}a^{\dagger2}a^2
 +F e^{-i\omega_dt}a^\dagger
 +\overline F e^{i\omega_dt}a.
 \label{eq:H-laboratory}
\end{equation}
If \(U(t)=e^{i\omega_dtN}\) and
\(\lvert\psi_{\mathrm{rot}}\rangle
=U(t)\lvert\psi_{\mathrm{lab}}\rangle\), then
\begin{equation}
 H_{\mathrm{rot}}
 =UH_{\mathrm{lab}}U^\dagger+i\hbar\dot U U^\dagger
 =\frac{V}{2}a^{\dagger2}a^2
  +\hbar(\omega_c-\omega_d)N
  +Fa^\dagger+\overline F a.
 \label{eq:rotating-frame-map}
\end{equation}
Thus the analysis is performed in the frame rotating at the drive frequency
\(\omega_d\), and \eqref{eq:H} uses the signed drive--cavity detuning
\(\Delta:=\omega_c-\omega_d\).  Sources using the opposite convention
\(\omega_d-\omega_c\) have detuning \(-\Delta\) relative to ours.  For the laboratory model
\eqref{eq:H-laboratory}, the displayed rotation is exact.  When
\eqref{eq:H} is obtained instead from a coordinate-space Duffing Hamiltonian,
discarding counter-rotating terms is an additional rotating-wave
approximation \cite{PeanoThorwart2006Duffing}.

Accordingly, \(E\) below is an eigenvalue of the autonomous rotating-frame
operator \eqref{eq:H}.  Relative to the periodic laboratory problem it is a
Floquet quasienergy representative, defined modulo \(\hbar\omega_d\) and
subject to the chosen rotating-frame and additive-energy conventions.  These
physical interpretations do not alter the notation or spectral problem used
in the remainder of the manuscript.
Since \(a^{\dagger 2}a^2=N(N-1)\), we write
\begin{equation}
  H=H_0+W,\qquad
  H_0=\frac{V}{2}N(N-1)+\hbar\Delta N,\qquad
  W=Fa^\dagger+\overline F\,a.
  \label{eq:split}
\end{equation}

Under the Bargmann transform, \(\Fock\) is identified with
the Hilbert space of entire functions \cite{Bargmann1961Hilbert}
\begin{equation}
 \mathcal{F}_{\mathrm B}
 =\left\{f\colon\C\to\C\ \text{entire}:\
 \|f\|_{\mathrm B}^2=
 \int_{\C}|f(z)|^2e^{-|z|^2}\,\frac{d^2z}{\pi}<\infty\right\},
 \label{eq:Bargmann-space}
\end{equation}
and \(a^\dagger\) and \(a\) act as multiplication by \(z\) and differentiation,
respectively.  Thus the eigenvalue equation \(H\Psi=E\Psi\) becomes
\begin{equation}
 \frac{V}{2}z^2\Psi''(z)+(\hbar\Delta z+\overline F)\Psi'(z)
 +(Fz-E)\Psi(z)=0,
 \label{eq:Bargmann-ode}
\end{equation}
with the physical condition \(\Psi\in\mathcal{F}_{\mathrm B}\).

\section{Basic operator-theoretic properties}
\label{sec:operator-properties}

In a thorough study whose main result is the absence of Bose--Einstein
condensation in one and two dimensions, Stachura, Pusz and Wojtkiewicz first
establish self-adjointness and lower semiboundedness for the Bose--Hubbard
Hamiltonian on a finite lattice \cite{StachuraPuszWojtkiewicz2020BoseHubbard}.
Their grand-canonical Hamiltonian, written
\(H_\Lambda(u)=u\hat N_2+T_\Lambda-\mu\hat N+\lambda L_\Lambda\) in their
Eq.~(16), contains the on-site interaction term
\(\hat N_2=\sum_x\hat n_x^2\), the hopping term \(T_\Lambda\), and the
U(1)-symmetry-breaking term \(L_\Lambda=\sum_x(c_x^\dagger+c_x)\) --- which is
exactly our drive term, named there from a thermodynamic rather than a
quantum-optical point of view.  Their Theorem~3 establishes: essential
self-adjointness on the finite-particle subspace; that the domain of the
resulting closure equals the domain of the interaction operator; an explicit
lower bound; and trace-class \(\exp(-\beta H)\) for every \(\beta>0\).  The
proof proceeds, as in our Lemma~\ref{lem:relative-bound}, by showing that
each of \(T'\), \(T''\),
\(\hat N\), and \(L\) is \(\hat N_2\)-bounded with relative bound zero (their
Proposition~11).  Consequently, the results of the present section are the
\(|\Lambda|=1\) case of their theorem, and trace-class \(\exp(-\beta H)\) is a
stronger statement than compactness of the resolvent, so it also subsumes our
Theorem~\ref{thm:compact-resolvent}.  In the dictionary between the two
treatments, our single site corresponds to \(|\Lambda|=1\); their
\(u\hat N_2\) and our \((V/2)N(N-1)\) differ by a term linear in \(N\),
absorbed into the detuning; and their \(\lambda\) is real, while our \(F\) is
complex and reduced to a real coupling by a phase rotation of the number
operator, a rotation used again in the weak-drive analysis of
Section~\ref{sec:weak-drive} below.  We nonetheless give a self-contained
proof for the single-mode case, since it produces explicit
constants matched to this case and is phrased in the Bargmann representation
used throughout the rest of the paper.

We first isolate the elementary relative-boundedness estimate used below.

\begin{lemma}\label{lem:relative-bound}
The driving term \(W\) is infinitesimally bounded with respect to \(H_0\): for
every \(\varepsilon>0\) there exists \(C_\varepsilon<\infty\) such that
\begin{equation}
 \|W\psi\|\leq \varepsilon\|H_0\psi\|
       +C_\varepsilon\|\psi\|,
 \qquad \psi\in\Dom(N^2).
 \label{eq:infinitesimal-bound}
\end{equation}
\end{lemma}

\begin{proof}
The standard identities
\(
 \|a\psi\|=\|N^{1/2}\psi\|
\)
and
\(
 \|a^\dagger\psi\|=\|(N+\id)^{1/2}\psi\|
\)
show that \(W\) is bounded from \(\Dom((N+\id)^{1/2})\) to \(\Fock\).
For every \(\delta>0\), the scalar inequality
\((n+1)^{1/2}\leq\delta n^2+c_\delta\), \(n\in\Nzero\), and the spectral
calculus give
\begin{equation}
 \|(N+\id)^{1/2}\psi\|
 \leq \delta\|N^2\psi\|+c'_{\delta}\|\psi\|.
 \label{eq:Nhalf-Ntwo}
\end{equation}
Moreover, the real quadratic polynomial
\(p(n)=(V/2)n(n-1)+\hbar\Delta n\) has positive leading coefficient.  Hence
\(n^2\leq c_1|p(n)|+c_2\) on \(\Nzero\), and therefore
\begin{equation}
 \|N^2\psi\|\leq c_1\|H_0\psi\|+c_2\|\psi\|.
 \label{eq:Ntwo-Hzero}
\end{equation}
Combining these estimates, and then choosing \(\delta\) sufficiently small,
proves \eqref{eq:infinitesimal-bound}.
\end{proof}

\begin{proposition}[Self-adjointness]\label{prop:self-adjoint}
The operator \(H\) is self-adjoint on
\begin{equation}
 \Dom(H)=\Dom(H_0)=\Dom(N^2).
 \label{eq:domain}
\end{equation}
It is essentially self-adjoint on the finite-particle subspace
\(\Fock_{\mathrm{fin}}=\operatorname{span}\{|n\rangle:n\in\Nzero\}\).
\end{proposition}

\begin{proof}
The operator \(H_0=p(N)\) is self-adjoint on \(\Dom(N^2)\), while \(W\) is
symmetric there.  By Lemma~\ref{lem:relative-bound}, \(W\) is \(H_0\)-bounded
with relative bound zero.  The Kato--Rellich theorem
\cite{Kato1995Perturbation} therefore implies that
\(H=H_0+W\) is self-adjoint on \(\Dom(H_0)=\Dom(N^2)\).

The finite-particle subspace is a core for \(H_0\).  Since the graph norms of
\(H\) and \(H_0\) are equivalent under a perturbation of relative bound less
than one, it is also a core for \(H\).  This proves essential
self-adjointness of the restriction of \(H\) to \(\Fock_{\mathrm{fin}}\).
\end{proof}

\begin{proposition}[Lower semiboundedness]\label{prop:lower-bound}
There exists a constant \(C_{\mathrm{lb}}<\infty\), depending on \(V\),
\(\Delta\), and \(F\), such that
\begin{equation}
 \langle\psi,H\psi\rangle\geq-C_{\mathrm{lb}}\|\psi\|^2,
 \qquad \psi\in\Dom(N^2).
 \label{eq:lower-bound}
\end{equation}
\end{proposition}

\begin{proof}
For normalized \(\psi\in\Dom(N^2)\), the Cauchy--Schwarz inequality gives
\begin{equation}
 |\langle\psi,W\psi\rangle|
 \leq 2|F|\,\langle\psi,(N+\id)\psi\rangle^{1/2}.
 \label{eq:form-drive}
\end{equation}
Because \(V>0\), scalar comparison of polynomials gives constants
\(C_0,C_1<\infty\) such that
\begin{equation}
 H_0\geq \frac{V}{4}N^2-C_0\id,
 \qquad
 2|F|\sqrt{x+1}\leq \frac{V}{8}x^2+C_1
 \quad (x\geq0).
 \label{eq:scalar-lower}
\end{equation}
Explicitly, completing the square gives the constant
\(C_0=(\hbar\Delta-V/2)^2/V\), and Young's inequality (exponents \(4\) and
\(4/3\)) gives the constant
\(C_1=2|F|+\tfrac32|F|^{4/3}(4/V)^{1/3}\) (for \(F=0\) this reduces to the
trivial bound \(C_1=0\), consistent with the formula).
Moreover, Cauchy--Schwarz implies
\(\langle N\rangle^2\leq\langle N^2\rangle\).  Hence
\begin{align}
 \langle\psi,H\psi\rangle
 &\geq \frac{V}{4}\langle N^2\rangle-C_0
       -2|F|\sqrt{\langle N\rangle+1} \\
 &\geq \frac{V}{8}\langle N^2\rangle-C_0-C_1
 \geq -C_0-C_1=:-C_{\mathrm{lb}}.
\end{align}
Homogeneity gives \eqref{eq:lower-bound} for arbitrary \(\psi\).
\end{proof}

\begin{theorem}[Compact resolvent and discrete spectrum]
\label{thm:compact-resolvent}
The Hamiltonian \(H\) has compact resolvent.  Consequently, by the standard
spectral theory of self-adjoint operators with compact resolvent
\cite{ReedSimon1978Analysis}, its spectrum is purely discrete and may be
listed, with multiplicities, as
\begin{equation}
 E_0\leq E_1\leq E_2\leq\cdots,\qquad E_j\longrightarrow+\infty.
 \label{eq:discrete-spectrum}
\end{equation}
Every eigenvalue has finite multiplicity, and the spectrum has no finite
accumulation point.
\end{theorem}

\begin{proof}
The embedding \(\Dom(N^2)\hookrightarrow\Fock\), with \(\Dom(N^2)\) equipped
with its graph norm, is compact.  Indeed, if \(\psi=\sum_{n\geq0}c_n|n\rangle\)
belongs to a graph-norm bounded set, then
\begin{equation}
 \sum_{n>M}|c_n|^2
 \leq \frac{1}{(M+1)^4}\sum_{n>M}n^4|c_n|^2,
 \label{eq:tail-bound}
\end{equation}
so its tails tend to zero uniformly; finite-dimensional truncation then proves
relative compactness.

By Proposition~\ref{prop:self-adjoint}, \(\Dom(H)=\Dom(N^2)\), and the graph norms
of \(H\) and \(H_0\), hence also those of \(H\) and \(N^2\), are equivalent.
Thus the embedding \(\Dom(H)\hookrightarrow\Fock\) is compact.  For any
\(\lambda\) in the resolvent set of \(H\), the resolvent
\((H-\lambda)^{-1}\) maps \(\Fock\) boundedly into \(\Dom(H)\).  Composing it
with the compact embedding proves that \((H-\lambda)^{-1}\) is compact.
The spectral assertions now follow from self-adjointness, compactness of the
resolvent, and Proposition~\ref{prop:lower-bound}.
\end{proof}

\begin{remark}
Theorem~\ref{thm:compact-resolvent} proves spectral quantization without
diagonalizing \(H\) and without solving \eqref{eq:Bargmann-ode}.  In the
Bargmann picture, entire analyticity alone should not be confused with the
physical boundary condition: an eigenfunction must also satisfy the Gaussian
square-integrability condition \eqref{eq:Bargmann-space}.  The relation between
this global growth condition and a characteristic equation for \(E\) will be
developed in the sequel.
\end{remark}

\begin{remark}[Scope of the single-mode results]
The multi-site extension of Proposition~\ref{prop:self-adjoint},
Proposition~\ref{prop:lower-bound}, and Theorem~\ref{thm:compact-resolvent},
including both the hopping and the drive terms, is already established for
arbitrary finite lattices in \cite{StachuraPuszWojtkiewicz2020BoseHubbard}.
What this section contributes beyond that result is the explicit constants
\(C_0\), \(C_1\), and \(C_{\mathrm{lb}}\) produced directly for the
single-mode case, together with the formulation in the Bargmann
representation used throughout the rest of the paper.  The proof given here
is self-contained and does not rely on the lattice apparatus of
\cite{StachuraPuszWojtkiewicz2020BoseHubbard}.
\end{remark}

\section{Whittaker--Ince characteristic values}
\label{sec:heunwi}

Multiplying \eqref{eq:Bargmann-ode} by \(2/V\) identifies it with the
Whittaker--Ince equation
\begin{equation}
 z^2U''+(B_1+B_2z)U'+(B_3+qz)U=0,
 \label{eq:wi-equation}
\end{equation}
under the map
\begin{equation}
 B_1=\frac{2\overline F}{V},\qquad
 B_2=\frac{2\hbar\Delta}{V},\qquad
 B_3=-\frac{2E}{V},\qquad
 q=\frac{2F}{V}.
 \label{eq:wi-kerr-map}
\end{equation}
The genealogy of \eqref{eq:wi-equation} as a singular limit of the
double-confluent Heun equation is verified in
Section~\ref{sec:liouville}.  We first define the special-function objects
needed for the spectral statement.

\paragraph{Author-defined notation.}
The symbol \(\operatorname{HeunWI}\) introduced below is a definition
specific to this manuscript.  It is not notation of Figueiredo, El-Jaick,
Maple, Wolfram Language, or DLMF.

For \(B_1\ne0\), define
\begin{equation}
 \operatorname{HeunWI}(B_1,B_2,B_3,q;z)
 \label{eq:heunwi-formal}
\end{equation}
initially as the unique normalized \emph{formal} solution in
\(\C[[z]]\) of \eqref{eq:wi-equation} with constant coefficient \(U(0)=1\).
The origin is an irregular singular point, so this definition does not assert
that the formal series converges.  If it does converge, its sum is the
corresponding holomorphic germ; if its radius is infinite, the same notation
denotes the resulting entire function.  No phase convention is imposed on
the independent complex parameters.  On the physical slice,
\(q=\overline{B_1}\) and \(B_2\in\mathbb R\), as follows from
\eqref{eq:wi-kerr-map}.

The condition \(B_1\ne0\) is essential to this normalization: when \(B_1=0\),
the constant term of the equation does not determine the next formal
coefficient, and \(U(0)=1\) is inconsistent unless \(B_3=0\); further
degeneracies then require a separate Frobenius or Euler analysis.  The case
\(q=0\) with \(B_1\ne0\) is allowed by the abstract definition and includes
terminating exceptional series.  In the physical family, however, \(q=0\)
and \(B_1=0\) occur together and are treated separately below.

For \(B_1\ne0\), coefficient extraction directly from \eqref{eq:wi-equation}
gives the three-term recurrence that determines the formal object
\(\operatorname{HeunWI}\) explicitly.  Set
\begin{equation}
 A_k=k(k-1)+B_2k+B_3.
 \label{eq:taylor-Ak}
\end{equation}
Direct substitution into \eqref{eq:wi-equation} then gives
\begin{equation}
 B_1(k+1)c_{k+1}+A_kc_k+qc_{k-1}=0,\qquad
 c_{-1}=0,\quad c_0=1.
 \label{eq:taylor-general-recurrence}
\end{equation}
This recurrence determines every \(c_k\) uniquely from \(c_0=1\), proving
uniqueness of the normalized formal object.

\subsection{Derived recurrence and computational characteristic equation}

On the physical slice, put
\begin{equation}
 d_k=\frac V2k(k-1)+\hbar\Delta k,\qquad g=|F|^2.
 \label{eq:dk-g}
\end{equation}
Substitution of \(\Psi(z)=\sum_{k\ge0}c_kz^k\) into
\eqref{eq:Bargmann-ode} gives
\begin{equation}
 \overline F(k+1)c_{k+1}+(d_k-E)c_k+Fc_{k-1}=0,
 \qquad c_{-1}=0,
 \label{eq:taylor-recurrence}
\end{equation}
and its first equation is
\begin{equation}
 \overline F c_1=Ec_0.
 \label{eq:origin-compatibility}
\end{equation}
This is the physical-slice specialization of the general recurrence
\eqref{eq:taylor-general-recurrence} under the substitution
\eqref{eq:wi-kerr-map}, and inherits its uniqueness conclusion for
\(c_0=1\).

\begin{lemma}[Asymptotic dichotomy for the physical-slice recurrence]
\label{lem:asymptotic-dichotomy}
Fix \(V>0\), \(\Delta\in\mathbb R\), \(F\ne0\), and \(E\in\mathbb C\), and
let \((c_k)_{k\ge0}\) be a nonzero solution of \eqref{eq:taylor-recurrence}
with \(c_k\ne0\) for all large \(k\).  Then exactly one of the following
holds as \(k\to\infty\):
\begin{itemize}
\item[(a)] \(c_{k+1}/c_k=-\dfrac{V}{2\overline F}\,k\,(1+o(1))\)
 (\emph{dominant} behavior); or
\item[(b)] \(c_k/c_{k-1}=-\dfrac{2F}{V}\,\dfrac1{k^2}\,(1+o(1))\)
 (\emph{minimal} behavior).
\end{itemize}
\end{lemma}

\begin{proof}
We invoke the classification of normalized three-term recurrences
\(c_{k+1}+u_kc_k+v_kc_{k-1}=0\) with \(u_k\sim P\,k^\alpha\) and
\(v_k\sim Q\,k^\beta\) as \(k\to\infty\), \(P,Q\ne0\)
\cite[Theorem~2.3(a)]{Gautschi1967Recurrence}: in the nondegenerate case
\(\alpha>\beta/2\), the recurrence admits, up to scale, a unique minimal
solution with ratio \(c_k/c_{k-1}\sim-(Q/P)\,k^{\beta-\alpha}\), and every
solution not proportional to it has ratio \(c_{k+1}/c_k\sim-P\,k^\alpha\).

Dividing \eqref{eq:taylor-recurrence} by \(\overline F(k+1)\) (valid since
\(F\ne0\)) puts it in this normalized form with
\begin{equation}
 u_k=\frac{d_k-E}{\overline F(k+1)},\qquad
 v_k=\frac F{\overline F(k+1)}.
\end{equation}
Since \(d_k-E=\tfrac V2k^2+(\hbar\Delta-\tfrac V2)k-E\), we read off
\(u_k\sim\dfrac V{2\overline F}\,k\), i.e.\ \(\alpha=1\),
\(P=V/(2\overline F)\), with \(P\ne0\) since \(V>0\); and
\(v_k\sim\dfrac F{\overline F}\,k^{-1}\), i.e.\ \(\beta=-1\),
\(Q=F/\overline F\), with \(Q\ne0\) since \(F\ne0\).  Then
\begin{equation}
 \alpha-\frac\beta2=1-\left(-\frac12\right)=\frac32>0
\end{equation}
for every \(V>0\), \(\Delta\in\mathbb R\), \(E\in\mathbb C\), and
\(F\ne0\): the nondegenerate case holds unconditionally, and the only
hypothesis used is \(F\ne0\).  The two ratios are
\begin{equation}
 -P\,k^\alpha=-\frac V{2\overline F}k,
 \qquad
 -\frac QP\,k^{\beta-\alpha}
 =-\frac F{\overline F}\cdot\frac{2\overline F}V\,k^{-2}
 =-\frac{2F}Vk^{-2},
\end{equation}
the second after cancelling \(\overline F\).  By the classification, the
minimal branch is unique up to scale with ratio
\(c_k/c_{k-1}\sim-(2F/V)k^{-2}\) --- case (b) --- and every solution not
proportional to it has ratio \(c_{k+1}/c_k\sim-(V/(2\overline F))k\) ---
case (a).
\end{proof}

\begin{theorem}[Continued-fraction characterization of the spectrum]
\label{thm:bargmann-characteristic}
For the physical family \eqref{eq:parameters} with \(F\ne0\), let \((c_k)\)
be the sequence determined by \eqref{eq:taylor-recurrence} with \(c_0=1\).
The following are equivalent:
\begin{itemize}
\item[(i)] the series \(\sum_{k\ge0}c_kz^k\) has infinite radius of
 convergence;
\item[(ii)] its sum belongs to the Bargmann space \(\mathcal F_{\rm B}\);
\item[(iii)] \((c_k)\) is the minimal solution, in the sense of
 Lemma~\ref{lem:asymptotic-dichotomy}, of \eqref{eq:taylor-recurrence};
\item[(iv)] \(E\) satisfies
\begin{equation}
 E-\cfrac{g}{
 E-d_1-\cfrac{2g}{
 E-d_2-\cfrac{3g}{
 E-d_3-\ddots}}}=0.
 \label{eq:characteristic-cf}
\end{equation}
\end{itemize}
We write \(\Xi(E)\) for the left-hand side of \eqref{eq:characteristic-cf},
so that condition (iv) reads \(\Xi(E)=0\).
The fraction in (iv) is the Pincherle limit belonging to the minimal
solution of \eqref{eq:taylor-recurrence}; at a pole the equation is
understood in its equivalent cross-multiplied form.  The set of \(E\)
satisfying these equivalent conditions is the spectrum of \(H\).  Every
eigenvalue is simple, with eigenfunction
\begin{equation}
 \Psi_n(z)=\mathcal N_n\,
 \operatorname{HeunWI}\!\left(
 \frac{2\overline F}{V},\frac{2\hbar\Delta}{V},
 -\frac{2E_n}{V},\frac{2F}{V};z\right);
 \label{eq:heunwi-eigenfunction}
\end{equation}
reality, discreteness, absence of a finite accumulation point, and
completeness follow from Proposition~\ref{prop:self-adjoint} and
Theorem~\ref{thm:compact-resolvent}.  These claims are restricted to the physical
slice; they are not asserted for arbitrary complex \(B_1,B_2,q\).
\end{theorem}

\begin{proof}
By Lemma~\ref{lem:asymptotic-dichotomy}, there is a unique minimal branch
up to scale, with leading order \(c_k^{\min}/c_{k-1}^{\min}\sim
-(2F/V)k^{-2}\), whereas every independent branch is dominant and
satisfies
\begin{equation}
 \frac{c_{k+1}^{\rm dom}}{c_k^{\rm dom}}
 \sim-\frac{V}{2\overline F}k.
 \label{eq:dominant-ratio}
\end{equation}
Given minimality, the subleading order in \(k\) follows from the
recurrence itself, by elementary manipulation: writing
\(s_k:=c_k/c_{k-1}\), \eqref{eq:taylor-recurrence} rearranges to the exact
relation \(s_k=-F/\bigl[(d_k-E)+\overline F(k+1)s_{k+1}\bigr]\).  By the
Lemma, \(s_{k+1}=O(k^{-2})\) for the minimal branch, so
\(\overline F(k+1)s_{k+1}=O(k^{-1})\), which is of relative order
\(k^{-3}\) against \(d_k-E\sim\tfrac V2k^2\); hence
\(s_k=-F/(d_k-E)\cdot\bigl(1+O(k^{-3})\bigr)\).  Expanding
\(d_k-E=\tfrac V2k^2\bigl[1+(B_2-1)/k+O(k^{-2})\bigr]\) and inverting the
bracket gives
\begin{equation}
 \frac{c_k^{\min}}{c_{k-1}^{\min}}
 =-\frac{2F}{V}\frac1{k^2}
 \left[1+\frac{1-B_2}{k}+O(k^{-2})\right].
 \label{eq:minimal-ratio}
\end{equation}
This refines the leading order given by the Lemma under the minimality
already established; it is not an independent proof of minimality.
The forward formal series therefore has infinite radius exactly when the
origin condition selects the minimal branch, proving (i)\(\Leftrightarrow\)(iii).
Pincherle's theorem \cite[Theorem~1.1]{Gautschi1967Recurrence} gives
\begin{equation}
 \frac{c_1^{\min}}{c_0^{\min}}
 =\frac{F}{
 E-d_1-\cfrac{2g}{
 E-d_2-\cfrac{3g}{
 E-d_3-\ddots}}},
 \label{eq:minimal-first-ratio}
\end{equation}
whose compatibility with \eqref{eq:origin-compatibility} is precisely
\eqref{eq:characteristic-cf}, proving (iii)\(\Leftrightarrow\)(iv).
Multiplying \eqref{eq:minimal-ratio} yields
\begin{equation}
 c_k^{\min}
 =\kappa_{\min}\,\frac{(-2F/V)^k}{(k!)^2}
 k^{1-B_2}\left[1+O(k^{-1})\right].
 \label{eq:minimal-coefficients}
\end{equation}
Thus the entire series has order \(1/2\), type
\(2\sqrt{2|F|/V}\), and
\(\sum_{k\ge0}k!|c_k|^2<\infty\), so the minimal branch belongs to
\(\mathcal F_{\rm B}\).  Conversely, a Bargmann coefficient sequence cannot
contain the factorially growing dominant branch, so (ii)\(\Leftrightarrow\)(iii)
as well, completing the equivalence of (i)--(iv).  The set of \(E\)
satisfying these equivalent conditions is, by definition, the set of
eigenvalues of \(H\) in the Bargmann representation, hence the spectrum of
\(H\); reality, discreteness, absence of a finite accumulation point, and
completeness follow from Proposition~\ref{prop:self-adjoint} and
Theorem~\ref{thm:compact-resolvent}.  Simplicity follows because for \(F\ne0\) the
coefficient of \(z^0\) fixes the first derivative from the value at the
origin, while a solution with zero value there vanishes identically.  These
claims are restricted to the physical slice; they are not asserted for
arbitrary complex \(B_1,B_2,q\).
\end{proof}

\begin{remark}[Undriven spectral problem]
For \(F=0\), neither the preceding normalization nor any formula obtained by
division by \(F\), \(\overline F\), \(B_1\), or \(q\) applies.  Directly,
\eqref{eq:Bargmann-ode} is an Euler equation.  Substitution of
\(\Psi=z^\lambda\) gives
\[
 \frac V2\lambda(\lambda-1)+\hbar\Delta\lambda-E=0.
\]
A nonzero single-valued entire Bargmann solution requires
\(\lambda=n\in\Nzero\).  Since
\(\|z^n\|_{\rm B}^2=n!\), the normalized eigenfunctions and energies are
\begin{equation}
 \Psi_n^{(0)}(z)=\frac{z^n}{\sqrt{n!}},\qquad
 E_n^{(0)}=\frac V2n(n-1)+\hbar\Delta n.
 \label{eq:undriven-spectrum}
\end{equation}
For \(n>0\), \(\Psi_n^{(0)}(0)=0\), so the driven normalization \(c_0=1\)
cannot be continued to this exactly degenerate point.  The forward Taylor
recurrence divided by \(B_1\) and the square-root expansion divided by
\(\sqrt{-q}\) both cease to exist.  Thus direct substitution into formulas
derived for \(B_1q\ne0\) is invalid; the exact boundary is distinct from
studying the physical limit \(F\to0\) with \(\overline F=F^*\).
\end{remark}

No novelty is claimed for the continued-fraction technique itself.
The \(F\to0\) reduction of \eqref{eq:characteristic-cf} is singular: roots
near \(d_n\), \(n\ge1\), emerge
through successive tail poles rather than by setting \(g=0\) only at the
outermost level.  Relative perturbation of the diagonal operator gives the
controlled multiset limit \(\{E_n\}\to\{d_n\}\), including
\eqref{eq:undriven-spectrum}, with possible relabelling at degeneracies.

\subsection{Taylor expansion at the origin}

For any possible Pad\'e and Hermite--Pad\'e constructions (not considered
in this work), it is useful to record the low-order behaviour of the
recurrence
\eqref{eq:taylor-general-recurrence}, with \(A_k\) as in
\eqref{eq:taylor-Ak}.  Writing \(A_1=B_2+B_3\), \(A_2=2+2B_2+B_3\), and
\(A_3=6+3B_2+B_3\), the first coefficients are
\begin{align}
 c_1={}&-\frac{B_3}{B_1},\nonumber\\
 c_2={}&\frac{B_3A_1-qB_1}{2B_1^2},\nonumber\\
 c_3={}&\frac{qB_1(A_2+2B_3)-B_3A_1A_2}{6B_1^3},\nonumber\\
 c_4={}&\frac{A_3B_3A_1A_2-qB_1A_3(A_2+2B_3)
 -3qB_1B_3A_1+3q^2B_1^2}{24B_1^4}.
 \label{eq:taylor-first-four}
\end{align}
Consequently,
\begin{equation}
 \operatorname{HeunWI}(B_1,B_2,B_3,q;z)
 =1+c_1z+c_2z^2+c_3z^3+c_4z^4+O(z^5).
 \label{eq:heunwi-taylor-four}
\end{equation}

On the physical slice, let \(R_k=E-d_k\), with \(d_k\) and \(g=|F|^2\)
as in \eqref{eq:dk-g}.  Formula~\eqref{eq:taylor-first-four} becomes
\begin{align}
 c_1={}&\frac{E}{\overline F},\nonumber\\
 c_2={}&\frac{R_1E-g}{2\overline F^{\,2}},\nonumber\\
 c_3={}&\frac{R_2(R_1E-g)-2gE}{6\overline F^{\,3}},\nonumber\\
 c_4={}&\frac{R_3[R_2(R_1E-g)-2gE]-3g(R_1E-g)}
 {24\overline F^{\,4}}.
 \label{eq:taylor-physical-four}
\end{align}
Direct symbolic substitution cancels the coefficients through \(z^3\);
the omitted \(c_5z^5\) first enters the residual at order \(z^4\).
The routine \texttt{taylor\_coefficients} in the accompanying local-series
module implements
\eqref{eq:taylor-general-recurrence} to arbitrary order without converting
the caller's exact, symbolic, complex, or arbitrary-precision scalar type.
It explicitly rejects \(B_1=0\) and invalid orders.

\subsection{Historical scope of continued fractions}

The verified early quantum-optics uses concern a related but different
problem.  Drummond and Walls solve the stationary generalized-\(P\)
Fokker--Planck equation for an open, coherently driven Kerr cavity
\cite{DrummondWalls1980Bistability}.  Vogel and Risken applied matrix
continued fractions to its quantum Fokker--Planck equation in 1987
\cite{VogelRisken1987Matrix}.  They subsequently expanded the open master
equation in Fourier--Laguerre modes and used a
\emph{matrix}-valued continued fraction: their recurrence (4.13), transfer
relation (5.4), and determinant condition (5.7) compute stationary
quasidistributions and Liouvillian decay rates, not eigenenergies of the
closed Hamiltonian \cite{VogelRisken1988Tunneling}.  Their later extension
again treats quasidistributions and the lowest nonzero master-equation
eigenvalue \cite{VogelRisken1989Quasiprobability}.  Haake, Risken, Savage,
and Walls derive a damped nonlinear-oscillator master equation and its
stationary state, not the scalar fraction \eqref{eq:characteristic-cf}
\cite{HaakeEtAl1986Master}.

None of these inspected formulas is algebraically identical to
\eqref{eq:characteristic-cf}: the older recurrences are block/matrix
recurrences for a density operator or phase-space distribution and include
damping.  They do establish important methodological prior art.  The present
backward search is not exhaustive, so absence of an inspected exact match is
not a priority claim; the result here is stated only as a HeunWI
interpretation and derivation for the closed spectral problem.

In summary, this section identifies the Bargmann eigenvalue equation
\eqref{eq:Bargmann-ode} with the Whittaker--Ince equation
\eqref{eq:wi-equation} and gives its normalized formal solution
\(\operatorname{HeunWI}\) explicitly by the three-term recurrence
\eqref{eq:taylor-general-recurrence}.  On the physical slice with
\(F\ne0\), Theorem~\ref{thm:bargmann-characteristic} shows that
convergence of the resulting series to a Bargmann eigenfunction,
selection of the minimal solution among the two asymptotic branches of
the recurrence, and satisfaction of the continued-fraction equation
\eqref{eq:characteristic-cf} are equivalent, and that their common
solution set is the spectrum of \(H\).  The case \(F=0\) is treated
separately above, and the \(F\to0\) limit recovers the unperturbed
energies \eqref{eq:undriven-spectrum}.

\section{Liouville normal form and the double-confluent Heun identification}
\label{sec:liouville}

We next remove the first derivative in \eqref{eq:Bargmann-ode}.  This is a
local transformation on a simply connected subdomain of \(\C\setminus\{0\}\),
on which a branch of \(\operatorname{Log} z\) has been fixed.  We work
directly in the parameters \(B_1,B_2,B_3,q\) already introduced by the map
\eqref{eq:wi-kerr-map} of Section~\ref{sec:heunwi}, rather than introducing a
second set of symbols for the same four quantities.  Dividing
\eqref{eq:wi-equation} by \(z^2\) --- equivalently, dividing
\eqref{eq:Bargmann-ode} by \((V/2)z^2\) and using \eqref{eq:wi-kerr-map} ---
the equation reads
\begin{equation}
 \Psi''+P(z)\Psi'+Q(z)\Psi=0,
 \qquad
 P(z)=\frac{B_2}{z}+\frac{B_1}{z^2},
 \qquad
 Q(z)=\frac{q}{z}+\frac{B_3}{z^2}.
 \label{eq:standard-ode}
\end{equation}

\begin{proposition}[Liouville normal form]
On the chosen branch domain, the substitution
\begin{equation}
 \Psi(z)=A(z)\Phi(z),
 \qquad
 A(z)=z^{-B_2/2}\exp\!\left(\frac{B_1}{2z}\right)
 \label{eq:liouville-substitution}
\end{equation}
transforms \eqref{eq:standard-ode} into
\begin{equation}
 \Phi''(z)+\left[
 -\frac{B_1^2}{4z^4}
 +\frac{B_1(1-B_2/2)}{z^3}
 +\frac{B_3+B_2/2-B_2^2/4}{z^2}
 +\frac{q}{z}
 \right]\Phi(z)=0.
 \label{eq:liouville-normal-form}
\end{equation}
\end{proposition}

\begin{proof}
Substitution of \(\Psi=A\Phi\) shows that the coefficient of \(\Phi'\) is
\(2A'/A+P\).  It vanishes precisely when
\begin{equation}
 \frac{A'}{A}=-\frac{P}{2},
 \qquad
 A=C\exp\!\left(-\frac12\int P(z)\,dz\right)
   =Cz^{-B_2/2}\exp\!\left(\frac{B_1}{2z}\right).
\end{equation}
The nonzero constant \(C\) may be absorbed into \(\Phi\), giving
\eqref{eq:liouville-substitution}.  The remaining coefficient is
\begin{equation}
 Q-\frac{P'}{2}-\frac{P^2}{4}.
\end{equation}
Inserting the functions in \eqref{eq:standard-ode} and collecting the powers
\(z^{-4},\ldots,z^{-1}\) gives \eqref{eq:liouville-normal-form}.
\end{proof}

The rank terminology and the separation of formal exponential and power
factors used here belong to the general asymptotic theory of singular linear
differential equations \cite{Olver1997Asymptotics,Fedoryuk1993Asymptotic}.
These citations provide background only: no theorem from either monograph is
being invoked here as already verified for the present ramified problem.

\begin{proposition}[Singularity classification]
\label{prop:singularity-classification}
Assume \(F\neq0\).  On the Riemann sphere, the only singular points of
\eqref{eq:standard-ode} are \(z=0\) and \(z=\infty\).  In scalar normal-form
terminology, the origin is an unramified irregular singular point of
Poincar\'e rank (1), while infinity is a ramified irregular singular point of
rank (1/2).  The latter becomes an unramified rank-(1) point after the
quadratic covering \(z=\zeta^2\).  The two leading formal behaviours at
infinity are
\begin{equation}
 \Phi_{\infty,\pm}(z)
 \sim z^{1/4}\exp\!\left(\pm2\sqrt{-qz}\right),
 \qquad
 \Psi_{\infty,\pm}(z)
 \sim z^{1/4-B_2/2}
       \exp\!\left(\pm2\sqrt{-qz}\right),
 \label{eq:infinity-formal-behaviours}
\end{equation}
on sectors with fixed branches of the square root.
\end{proposition}

\begin{proof}
At the origin, the coefficient in the normal-form equation
\eqref{eq:liouville-normal-form} has a pole of order (4), with nonzero leading
coefficient \(-B_1^2/4\).  The scalar rank is therefore
\((4-2)/2=1\).  To inspect infinity, put \(t=1/z\) and
\(Y(t)=\Psi(1/t)\).  Equation~\eqref{eq:standard-ode} becomes
\begin{equation}
 Y''(t)+\left(\frac{2-B_2}{t}-B_1\right)Y'(t)
 +\left(\frac{q}{t^3}+\frac{B_3}{t^2}\right)Y(t)=0.
 \label{eq:infinity-inverted}
\end{equation}
After removal of its first derivative, the coefficient has a pole of order
\(3\), because \(q\neq0\).  Hence the scalar rank is
\((3-2)/2=1/2\), and the half-integer rank accounts for the quadratic
ramification.

For the formal factors at infinity, substitute
\(\Psi=e^{\lambda\sqrt z}z^\rho(1+o(1))\) into
\eqref{eq:standard-ode}.  The coefficients of \(z^{-1}\) and \(z^{-3/2}\)
give
\begin{equation}
 \frac{\lambda^2}{4}+q=0,
 \qquad
 \lambda\left(\rho-\frac14+\frac{B_2}2\right)=0.
\end{equation}
Thus \(\lambda=\pm2\sqrt{-q}\) and
\(\rho=1/4-B_2/2\).  Since
\(A(z)=z^{-B_2/2}(1+O(z^{-1}))\) at infinity, division by (A) gives the
corresponding behaviours of \(\Phi\) in
\eqref{eq:infinity-formal-behaviours}.
\end{proof}

\begin{remark}[Undriven degeneration]
If \(F=0\), then \(B_1=q=0\), and both (0) and infinity are regular
singular points rather than irregular ones.  The Bargmann equation reduces to
the Euler equation
\begin{equation}
 \Psi''+\frac{B_2}{z}\Psi'+\frac{B_3}{z^2}\Psi=0,
\end{equation}
with indicial equation
\begin{equation}
 r(r-1)+B_2r+B_3=0.
\end{equation}
This degeneration is consistent with the fact that the undriven Hamiltonian
is diagonal in the number basis.
\end{remark}

The factor \(z^{-B_2/2}=\exp[-(B_2/2)\operatorname{Log} z]\) is multivalued unless
\(B_2/2\in\mathbb Z\); its monodromy about the origin is
\(e^{-\pi iB_2}\).  If \(F\neq0\), the exponential factor in \(A\) has an
essential singularity at \(z=0\), while remaining single-valued on the
punctured plane.  Thus \(A\) is an auxiliary local gauge factor, not a
Bargmann--Fock wavefunction.  The branch dependence introduced by \(A\) must
cancel in the product \(A\Phi\) whenever that product is the physical entire
function \(\Psi\).

\begin{proposition}[Formal behaviours at the origin]
Assume \(F\neq0\), so that \(B_1\neq0\).  The two leading formal local
behaviours admitted by \eqref{eq:liouville-normal-form} are
\begin{align}
 \Phi_{\mathrm{reg}}(z)
 &\sim z^{B_2/2}\exp\!\left(-\frac{B_1}{2z}\right),
 \label{eq:phi-regular-formal}\\
 \Phi_{\mathrm{sing}}(z)
 &\sim z^{2-B_2/2}\exp\!\left(\frac{B_1}{2z}\right).
 \label{eq:phi-singular-formal}
\end{align}
These expressions specify only the leading formal factors; sectorial
asymptotic existence and Stokes connection data are not asserted here.
\end{proposition}

\begin{proof}
Insert \(\Phi=e^{s/z}z^\rho(1+o(1))\) formally into
\eqref{eq:liouville-normal-form}.  The coefficients of \(z^{-4}\) and
\(z^{-3}\) give
\begin{equation}
 s^2-\frac{B_1^2}{4}=0,
 \qquad
 2s(1-\rho)+B_1\left(1-\frac{B_2}2\right)=0.
\end{equation}
For \(s=-B_1/2\) one obtains \(\rho=B_2/2\), whereas for
\(s=B_1/2\) one obtains \(\rho=2-B_2/2\).
\end{proof}

Returning to \(\Psi=A\Phi\) displays the decisive compensation mechanism:
\begin{align}
 A(z)\Phi_{\mathrm{reg}}(z)&\sim 1,
 \label{eq:compensated-branch}\\
 A(z)\Phi_{\mathrm{sing}}(z)&\sim
 z^{2-B_2}\exp\!\left(\frac{B_1}{z}\right).
 \label{eq:enhanced-branch}
\end{align}
In the first branch both the essential exponential and the fractional power
cancel.  This is compatible with the locally analytic Bargmann solution.  In
the second branch the essential exponential is enhanced rather than cancelled
and a power factor remains.  Equations \eqref{eq:compensated-branch} and
\eqref{eq:enhanced-branch} are formal local statements; by themselves they do
not decide global Bargmann--Fock membership and impose no condition on \(E\).

For completeness, we record a precise special-function classification without
using it as a spectral assertion.  The normalized DLMF convention for the
double-confluent Heun equation assumes a nonzero quadratic-derivative
coefficient, while B\"uhring's connection theory treats a generic equation with
two rank-one irregular endpoints \cite{DLMF31,Buhring1994DCHE}.  Analogous
difficulties have been encountered for the same equation by El-Jaick and
Figueiredo, who obtain convergent expansions at both singular points but
report finding no transformation connecting them
\cite{ElJaickFigueiredo2013Coulomb}.  Neither result therefore applies
automatically to the ramified specialization below, for which that
coefficient vanishes.  Define the broader canonical family by the explicit
convention
\begin{equation}
 z^2u''+(\gamma_{\rm D}+\delta_{\rm D}z+\varepsilon_{\rm D}z^2)u'
 +(\alpha_{\rm D}z-q_{\rm D})u=0.
 \label{eq:dche-canonical}
\end{equation}
Then \eqref{eq:standard-ode} is exactly the specialization
\begin{equation}
 \gamma_{\rm D}=\frac{2\overline F}{V},\qquad
 \delta_{\rm D}=\frac{2\hbar\Delta}{V},\qquad
 \varepsilon_{\rm D}=0,\qquad
 \alpha_{\rm D}=\frac{2F}{V},\qquad
 q_{\rm D}=-B_3=\frac{2E}{V}.
 \label{eq:dche-correspondence}
\end{equation}
This algebraic correspondence rigorously places the Bargmann equation on a
degenerate parameter hypersurface of the double-confluent Heun family as
defined by \eqref{eq:dche-canonical}.  It does not, without further analysis,
identify a preferred global Heun solution, prove a connection formula, or
yield an energy-quantization condition.  The case \(F=0\) is more degenerate
and the two formal behaviours above do not apply; in that case the original
Hamiltonian is already diagonal in the number basis.

More precisely, in the Leaver--Figueiredo convention the parent DCHE is
\begin{equation}
 z^2U''+(B_1+B_2z)U'
 +(B_3-2\eta\omega z+\omega^2z^2)U=0.
 \label{eq:leaver-dche}
\end{equation}
The singular Whittaker--Ince limit
\begin{equation}
 \omega\longrightarrow0,\qquad
 \eta\longrightarrow\infty,\qquad
 2\eta\omega=-q\quad\text{fixed}
 \label{eq:whittaker-ince-limit}
\end{equation}
gives \eqref{eq:wi-equation}
\cite{Figueiredo2005Ince,ElJaickFigueiredo2008Solutions,
ElJaickFigueiredo2013Coulomb}.  Thus setting
\(\varepsilon_{\rm D}=0\) describes the resulting equation algebraically,
but is not literally the same parameter operation: the limit removes
\(\omega^2z^2U\) while retaining \(-2\eta\omega zU=qzU\).  It also changes
infinity to the rank-\(1/2\), ramified point whose formal balance was recorded
in \eqref{eq:infinity-formal-behaviours}.

This section removes the first-order derivative from the Bargmann equation,
proves the singularity classification at the origin and at infinity, and
derives the two formal local behaviours at each singular point together
with the compensation mechanism that distinguishes the regular from the
singular branch there.  It places the resulting normal form, via the
algebraic setting \(\varepsilon_{\rm D}=0\), as an exact degeneration of
the double-confluent Heun family
\eqref{eq:dche-canonical}--\eqref{eq:dche-correspondence}, and, via the
Leaver--Figueiredo correspondence \eqref{eq:leaver-dche} and the
Whittaker--Ince limit \eqref{eq:whittaker-ince-limit}, separately
identifies the Whittaker--Ince equation as the outcome of a distinct
limiting operation converging to the same equation.
None of these local and formal results by itself selects a preferred global
solution, proves a connection formula, or imposes a condition on the
energy.

\section{Sectorial solutions at infinity}
\label{sec:sectorial-infinity}

\subsection{Formal expansion at infinity}

The leading factors in \eqref{eq:infinity-formal-behaviours} extend to a
formal series in half-integer powers.  To derive it directly from
\eqref{eq:wi-equation}, put \(z=t^2\) and \(Y(t)=\Psi(t^2)\).  Since
\[
 \Psi'=\frac{Y'}{2t},\qquad
 \Psi''=\frac{Y''}{4t^2}-\frac{Y'}{4t^3},
\]
the transformed equation is
\begin{equation}
 Y''+\left(\frac{2B_2-1}{t}+\frac{2B_1}{t^3}\right)Y'
 +\left(4q+\frac{4B_3}{t^2}\right)Y=0.
 \label{eq:wi-t-equation}
\end{equation}
An integer-power series in \(z^{-1}\) is not generally closed when
\(q\ne0\), because differentiation of the square-root exponential introduces
odd powers of \(t^{-1}=z^{-1/2}\).

Choose compatible square roots so that \(s=\sqrt{-q}\) and
\(\sqrt{-qz}=s\sqrt z=st\), and let
\begin{equation}
 a_\sigma=2\sigma s,\qquad
 p=\frac12-B_2,\qquad
 r=\frac p2=\frac14-\frac{B_2}{2},
 \qquad \sigma\in\{-1,1\},
 \label{eq:asymptotic-parameters}
\end{equation}
and insert
\[
 Y=e^{a_\sigma t}t^p\sum_{n\ge0}u_n^{(\sigma)}t^{-n},
 \qquad u_0^{(\sigma)}=1.
\]
The coefficients of \(t^0\) and \(t^{-1}\) in the normalized
equation give \(a_\sigma^2+4q=0\) and, since \(a_\sigma\ne0\),
\(2p+2B_2-1=0\), respectively.
Define
\begin{equation}
 C_{\rm nf}=p^2+(2B_2-2)p+4B_3.
 \label{eq:asymptotic-C}
\end{equation}
With \(u_j^{(\sigma)}=0\) for \(j<0\), coefficient extraction gives the
arbitrary-order formal recurrence
\begin{equation}
 u_{m+1}^{(\sigma)}
 =\frac{[m(m+1)+C_{\rm nf}]u_m^{(\sigma)}
 +2a_\sigma B_1u_{m-1}^{(\sigma)}
 +2B_1(p-m+2)u_{m-2}^{(\sigma)}}
 {2a_\sigma(m+1)}.
 \label{eq:asymptotic-recurrence}
\end{equation}
In particular,
\begin{align}
 u_1^{(\sigma)}={}&\frac{C_{\rm nf}}{2a_\sigma},\nonumber\\
 u_2^{(\sigma)}={}&\frac{B_1}{2}
 +\frac{C_{\rm nf}(C_{\rm nf}+2)}{8a_\sigma^2},\nonumber\\
 u_3^{(\sigma)}={}&
 \frac{C_{\rm nf}(C_{\rm nf}+2)(C_{\rm nf}+6)}{48a_\sigma^3}
 +\frac{B_1(3C_{\rm nf}+6+4p)}{12a_\sigma}.
 \label{eq:asymptotic-first-three}
\end{align}
Thus, on a sector carrying fixed logarithms,
\begin{equation}
 \Psi_\sigma(z)\sim
 e^{2\sigma\sqrt{-qz}}z^{1/4-B_2/2}
 \left[1+u_1^{(\sigma)}z^{-1/2}
 +u_2^{(\sigma)}z^{-1}
 +u_3^{(\sigma)}z^{-3/2}+O(z^{-2})\right].
 \label{eq:asymptotic-three-terms}
\end{equation}
This is a formal expansion: neither convergence nor uniform error bounds are
asserted, and the factors \(t^p\) and \(z^r\) require the stated logarithm
branches unless their respective exponents are integers.  Independent
substitution cancels the residual through \(t^{-4}\) in the normalized
equation \eqref{eq:wi-t-equation}; after truncation at \(u_3\), its first
uncancelled order is \(t^{-5}\).  Equivalently, in the equation multiplied
by \(t^2\), the first uncancelled order is \(t^{-3}=z^{-3/2}\), which
determines \(u_4\).
The routine \texttt{asymptotic\_coefficients} in
\texttt{src/kerr\_heun/local\_series.py} implements
\eqref{eq:asymptotic-recurrence}; its explicit
\texttt{sqrt\_minus\_q} argument records the chosen branch.  Exact rational
inputs are checked exactly, while numerical inputs must satisfy
\[
 |s^2+q|\leq 10^{-12}\max\{1,|s^2|,|q|\}.
\]
For symbolic inputs a decisive exact zero test is honored; a genuinely
undecidable symbolic relation remains the caller's responsibility.

\paragraph{Branches and sectors.}
On a simply connected sector choose \(\operatorname{Log}(-qz)\) and
\(\operatorname{Log}z\), and define
\[
 \sqrt{-qz}=\exp\!\left[\frac12\operatorname{Log}(-qz)\right],
 \qquad z^r=\exp[r\operatorname{Log}z].
\]
For this branch, \(\Psi_\sigma\) is exponentially dominant where
\(\Re(\sigma\sqrt{-qz})>0\), subdominant where it is negative, and of equal
exponential magnitude on the rays where it vanishes.  We call
\begin{equation}
 \Re\sqrt{-qz}=0\quad\hbox{Stokes rays},\qquad
 \Im\sqrt{-qz}=0\quad\hbox{anti-Stokes rays}.
 \label{eq:stokes-ray-convention}
\end{equation}
This naming convention is not universal; the displayed equations, rather
than the names, define our usage.  No connection matrices or Stokes
multipliers are inferred here.

With the chosen value \(s=\sqrt{-q}\) held fixed, recurrence
\eqref{eq:asymptotic-recurrence} gives
\(u_n^{(-\sigma)}=(-1)^nu_n^{(\sigma)}\) when only
\(\sigma\mapsto-\sigma\).  The simultaneous change
\((\sigma,s)\mapsto(-\sigma,-s)\), by contrast, leaves
\(a_\sigma=2\sigma s\) and the coefficient sequence unchanged.

More precisely, let continuation from \(t\) to \(-t\) give
\(\operatorname{Log}(-t)=\operatorname{Log}t+\varepsilon i\pi\), where
\(\varepsilon=+1\) or \(-1\) according to the chosen half-turn.  Termwise,
\((-t)^p=e^{\varepsilon i\pi p}t^p\) and
\((-t)^{-n}=(-1)^nt^{-n}\).  Consequently the normalized formal expressions
satisfy the exact sheet-exchange rule
\begin{equation}
 Y_\sigma(-t)=e^{\varepsilon i\pi p}Y_{-\sigma}(t)
 \quad\text{formally},
 \qquad p=\frac12-B_2,\quad \varepsilon\in\{-1,1\},
 \label{eq:formal-sheet-exchange}
\end{equation}
where the left side means continuation along that half-turn.  Thus equality
without the phase is not generally valid.  This rule concerns only formal
expressions on the cover; it asserts neither the existence of sectorial
solutions with these expansions nor Bargmann admissibility or a global
spectral condition.

Under the analytic continuation \(z\mapsto ze^{2\pi i}\), one
counterclockwise circuit of \(z\) around the origin in the finite
\(z\)-plane, \(\sqrt z\mapsto-\sqrt z\),
\(z^r\mapsto e^{2\pi ir}z^r\), and
\(z^{-n/2}\mapsto(-1)^nz^{-n/2}\).  Thus the two exponential expressions
are interchanged (equivalently relabelled with fixed \(s\)), while the
half-integer formal series acquires its termwise factors \((-1)^n\), apart
from the algebraic monodromy factor.

\subsection{Formal Liouville--Green approximation at infinity}
\label{subsec:formal-liouville-green}

We now connect the direct expansion above with the ordinary
Liouville--Green approximation.  Recall that
\begin{equation}
 q=\frac{2F}{V}\in\C.
 \label{eq:lg-q-definition}
\end{equation}
The drive \(F\), and hence \(q\), may be genuinely complex; neither is
assumed real or positive.  On the sector under consideration, fix branches
of \(\sqrt q\) and \(\sqrt z\), and set
\begin{equation}
 s\coloneqq i\sqrt q,
 \qquad s^2=-q.
 \label{eq:lg-s-definition}
\end{equation}
For real \(V\), the physical relation is
\(B_1=2F^*/V\), and therefore \(B_1=q^*\).  In particular, \(B_1\), \(q\),
and \(s\) need not be real.  Write
\begin{equation}
 D=B_3+\frac{B_2}{2}-\frac{B_2^2}{4},\qquad
 \mathcal{H}=B_1\left(1-\frac{B_2}{2}\right).
 \label{eq:lg-DH}
\end{equation}
The exact Liouville transformation is
\begin{equation}
 \Psi=z^{-B_2/2}\exp\!\left(\frac{B_1}{2z}\right)\Phi,
 \qquad
 \Phi''+R(z;E)\Phi=0,
 \end{equation}
where, consistently with \eqref{eq:liouville-normal-form},
\begin{equation}
 R(z;E)=\frac qz+\frac D{z^2}+\frac{\mathcal{H}}{z^3}
        -\frac{B_1^2}{4z^4}.
 \label{eq:lg-normal-form}
\end{equation}
It is meromorphic on the finite plane, with its only finite singularity at
the origin (generically a fourth-order pole), and tends to zero as
\(q/z+O(z^{-2})\) at complex infinity.  The exceptional cancellations at
the origin and the undriven degeneration \(q=0\) are excluded below.

Fix a simply connected sector at infinity compatible with the selected
branches, and take \(q\ne0\), hence \(s\ne0\).  With
\(\sqrt{-qz}=s\sqrt z\), define
\begin{equation}
 p(z;E)=\sqrt{-R(z;E)},\qquad
 p(z;E)\sim sz^{-1/2}\sum_{j\geq0}c_jz^{-j},\qquad c_0=1.
 \label{eq:lg-momentum}
\end{equation}
The square root is the one asymptotic to \(sz^{-1/2}\).  Direct expansion
gives
\begin{align}
 c_1={}&\frac{D}{2q},\nonumber\\
 c_2={}&\frac{\mathcal{H}}{2q}-\frac{D^2}{8q^2},\nonumber\\
 c_3={}&-\frac{B_1^2}{8q}-\frac{DH}{4q^2}
          +\frac{D^3}{16q^3}.
 \label{eq:lg-momentum-coefficients}
\end{align}
In particular,
\begin{equation}
 \int^z p(t;E)\,dt
 \sim 2s\sqrt z+\frac{D}{s}z^{-1/2}
 -\frac{s c_2}{3}z^{-3/2}+\cdots.
 \label{eq:lg-phase}
\end{equation}
Moreover,
\begin{equation}
 p^{-1/2}\sim s^{-1/2}z^{1/4}
 \left(1-\frac{D}{4q}z^{-1}+\cdots\right).
 \label{eq:lg-phase-prefactor}
\end{equation}
After absorbing the constant \(s^{-1/2}\), the lowest-order formula
\(\Phi_\sigma^{[0]}=p^{-1/2}\exp(\sigma\int^z p\,dt)\) therefore yields
\begin{equation}
 \Psi_\sigma^{[0]}(z)\sim
 e^{2\sigma s\sqrt z}z^{1/4-B_2/2}
 \left[1+\frac{\sigma D}{s}z^{-1/2}+O(z^{-1})\right].
 \label{eq:lg-leading-psi}
\end{equation}
This reproduces exactly the direct formal exponential rate and algebraic power.
For the fixed branch of \(\sqrt z\), the \(\sigma\) branch grows where
\(\Re(\sigma s\sqrt z)>0\), decays where this real part is negative, and has
equal exponential magnitude with the other branch where it vanishes.  These
sectors depend on \(\arg(F)=\arg(q)\).  Changing the chosen branch of
\(\sqrt q\), and hence sending \(s\mapsto-s\), merely interchanges the labels
\(\sigma=+1\) and \(\sigma=-1\) of the two formal solutions.
The leading Liouville--Green approximation does not reproduce the complete
first series coefficient.  Indeed, since
\begin{equation}
 C_{\rm nf}=4D-\frac34,
 \end{equation}
the direct result is
\(u_1^{(\sigma)}=D/(\sigma s)-3/(16\sigma s)\), whereas leading WKB gives
only \(D/(\sigma s)\).

The missing term is the first formal transport correction.  For the exact
logarithmic derivative \(S=\Phi'/\Phi\), the Riccati equation
\(S'+S^2=p^2\) gives
\begin{equation}
 S_\sigma\sim \sigma p-\frac{p'}{2p}
 +\sigma\left(\frac{p''}{4p^2}-\frac{3p'^2}{8p^3}\right)+\cdots.
 \label{eq:lg-riccati-transport}
\end{equation}
The last parenthesis and its contribution to the exponent are, respectively,
\begin{equation}
 \frac{3}{32s}z^{-3/2}+O(z^{-5/2}),
 \qquad -\frac{3\sigma}{16s}z^{-1/2}+O(z^{-3/2}).
\end{equation}
Hence
\begin{center}
\begin{tabular}{c|c|c|c}
 quantity & direct series & leading WKB & corrected WKB\\ \hline
 rate in \(\sqrt z\) & \(2\sigma s\) & \(2\sigma s\) & \(2\sigma s\)\\
 power of \(z\) & \(\frac14-\frac{B_2}{2}\) & same & same\\
 \(u_1^{(\sigma)}\) & \(\frac{D}{\sigma s}-\frac{3}{16\sigma s}\)
 & \(\frac{D}{\sigma s}\)
 & \(\frac{D}{\sigma s}-\frac{3}{16\sigma s}\)
\end{tabular}
\end{center}

Analytic continuation across a chosen cut changes the relevant logarithm and
may reverse the sign of \(p\), thereby relabelling the two signs.  On a fixed
sheet, the curves \(\Re\int^z p(t)dt=0\) have equal exponential magnitudes,
whereas on \(\Im\int^z p(t)dt=0\) the action is real; these equations are our
operative convention independent of competing Stokes terminology.  All
statements in this subsection are formal and sectorial.  Ordinary WKB is
singular at a turning point \(R(z;E)=0\), and no uniform validity or error
bound is claimed.  This failure motivates a subsequent Olver uniform
analysis.

\subsection{Uniform asymptotic expansions: Olver's approach}
\label{subsec:olver-uniform}

Put \(z=t^2\), with compatible logarithms, in
\eqref{eq:wi-t-equation}.  Direct substitution of
\begin{equation}
 \Psi(t^2)=t^{1/2-B_2}\exp\!\left(\frac{B_1}{2t^2}\right)w(t)
 \label{eq:olver-exact-substitution}
\end{equation}
gives the exact equation
\begin{equation}
 w''=\{a_\sigma^2+g(t)\}w,\qquad
 g(t)=-\frac{C_{\rm nf}}{t^2}-\frac{4\mathcal{H}}{t^4}+\frac{B_1^2}{t^6},
 \label{eq:olver-covering-normal-form}
\end{equation}
where \(a_\sigma=2\sigma s\), \(s^2=-q\), and
\begin{equation}
 D=B_3+\frac{B_2}{2}-\frac{B_2^2}{4},\quad
 \mathcal{H}=B_1\left(1-\frac{B_2}{2}\right),\quad C_{\rm nf}=4D-\frac34.
\end{equation}
Since \(a_\sigma^2=-4q\), the equation is independent of \(\sigma\).
Transforming the exact normal form \eqref{eq:lg-normal-form} after
\(z=t^2\) gives the same three coefficients and independently checks their
signs.  Infinity has become an unramified rank-1 point, and
\(g=O(t^{-2})\) is integrable on rays.

Fix \(\sigma\in\{+1,-1\}\).  Choose a unit direction \(d\) with
\(\Re(a_\sigma d)\leq0\) and a sufficiently large \(R>0\), and define the
simply connected half-plane
\begin{equation}
 \Omega_{d,R}=\{t:\ x_d(t):=\Re(\overline d t)>R\},\qquad
 \Gamma_d(t)=\{t+sd:s\geq0\}.
 \label{eq:olver-progressive-condition}
\end{equation}
Every translated half-ray stays in \(\Omega_{d,R}\), avoids zero, and,
oriented from \(t\) to infinity, satisfies
\(\Re\{a_\sigma(u-t)\}=s\Re(a_\sigma d)\leq0\).  Thus the inequality,
orientation, and normalization by \(e^{a_\sigma t}\) agree explicitly.

\begin{theorem}[Normalized sectorial solutions at infinity]
\label{thm:olver-sectorial}
Let \(q\ne0\), fix \(s^2=-q\), then fix
\(\sigma\in\{+1,-1\}\), and choose \(d,R\) in that order as above.  There
is a unique holomorphic solution \(w_{\sigma,d}\) on \(\Omega_{d,R}\)
\begin{equation}
 w_{\sigma,d}=e^{a_\sigma t}(1+\varepsilon_{\sigma,d}),\qquad
 \varepsilon_{\sigma,d},\varepsilon_{\sigma,d}'\longrightarrow0,
\end{equation}
as \(x_d(t)\to\infty\), uniformly in \(\Omega_{d,R}\).  If
\begin{equation}
 I_d(t)=\int_0^\infty|g(t+sd)|\,ds,\qquad
 G_d(t)=\frac{I_d(t)}{|a_\sigma|},
\end{equation}
then throughout that half-plane
\begin{align}
 |\varepsilon_{\sigma,d}(t)|&\leq e^{G_d(t)}-1,
 \label{eq:olver-epsilon-bound}\\
 |\varepsilon_{\sigma,d}'(t)|&\leq I_d(t)e^{G_d(t)},
 \label{eq:olver-derivative-bound}
\end{align}
\begin{equation}
 I_d(t)\leq M_d(t):=
 \frac{|C_{\rm nf}|}{x_d(t)}+\frac{4|\mathcal{H}|}{3x_d(t)^3}
 +\frac{|B_1|^2}{5x_d(t)^5}.
 \label{eq:olver-radial-majorant}
\end{equation}
The conclusion includes \(\Re(a_\sigma d)=0\): the kernel is still bounded,
the integrals converge, and the construction remains holomorphic.
The other sign generally requires a different progressive direction and a
different half-plane or sector.
\end{theorem}

\begin{proof}
Writing \(w_{\sigma,d}=e^{a_\sigma t}h_{\sigma,d}\) gives
\(h_{\sigma,d}''+2a_\sigma h_{\sigma,d}'=g h_{\sigma,d}\).  Variation of constants
with the stated endpoint conditions gives, after \(u=t+sd\),
\begin{equation}
 h_{\sigma,d}(t)=1+\int_{\Gamma_d(t)}K_\sigma(t,u)g(u)h_{\sigma,d}(u)\,du,
 \quad K_\sigma(t,u)=\frac{e^{2a_\sigma(u-t)}-1}{2a_\sigma}.
 \label{eq:olver-volterra}
\end{equation}
Here \(K(t,t)=0\), \(K_t=-e^{2a_\sigma(u-t)}\), and
\(K_{tt}+2a_\sigma K_t=0\); two differentiations recover the equation and
endpoint conditions.  By \eqref{eq:olver-progressive-condition},
\(|K_\sigma|\leq|a_\sigma|^{-1}\).  Moreover
\(|t+sd|\geq x_d(t)+s\), which proves \eqref{eq:olver-radial-majorant}.
Successive approximations are majorized by \(G_d^n/n!\), locally uniformly
on \(\Omega_{d,R}\).  Each iterate is holomorphic: this follows by
differentiation under the fixed-parameter integral in \(s\), whose integrand
and derivatives have locally uniform integrable majorants.  Weierstrass'
theorem therefore makes the limit holomorphic.  Differentiating the limiting
equation recovers the ODE and gives \eqref{eq:olver-derivative-bound}; the
same Volterra estimate for a difference proves uniqueness in the stated
normalized class.

Let \(\mathcal D\) be a connected closed arc of unit directions satisfying
\(\Re(a_\sigma d)\leq0\), and consider the union of the corresponding
half-planes \(\mathcal S_{\sigma,R}=\bigcup_{d\in\mathcal D}\Omega_{d,R}\).
They cover the associated proper progressive sector sufficiently far out.
For neighboring \(d_1,d_2\in\mathcal D\), require
\(\Re(\overline d_1d_2)>0\).  If
\(t\in\Omega_{d_1,R}\cap\Omega_{d_2,R}\), then for every \(s\geq0\)
\begin{align}
 x_{d_1}(t+sd_2)&=x_{d_1}(t)
   +s\Re(\overline d_1d_2)>R,\nonumber\\
 x_{d_2}(t+sd_2)&=x_{d_2}(t)+s>R.
 \label{eq:olver-overlap-rays}
\end{align}
Thus the complete \(d_2\)-ray remains in the overlap, where
\(h_{\sigma,d_1}\) is defined and holomorphic.  Moreover
\(x_{d_1}(t+sd_2)\to\infty\), so its normalization and derivative
normalization hold along that ray.  Variation of constants for the same ODE
along this complete ray therefore shows that \(h_{\sigma,d_1}\) satisfies
the \(d_2\)-Volterra equation.  Uniqueness for that equation gives
\(w_{\sigma,d_1}=w_{\sigma,d_2}\) throughout the overlap.  A finite chain
of directions in \(\mathcal D\), chosen with the same positive-inner-product
condition for consecutive members, gives all pairwise identifications.
Consequently the half-plane solutions glue to one holomorphic canonical
solution precisely on \(\mathcal S_{\sigma,R}\).
\end{proof}

For a radial path beginning at \(t=re^{i\theta}\), take
\(d=e^{i\theta}\).  Then \(x_d(t)=|t|\), so
\eqref{eq:olver-radial-majorant} becomes exactly
\(|C_{\rm nf}|/|t|+4|\mathcal{H}|/(3|t|^3)+|B_1|^2/(5|t|^5)\).  This radial formula is a
corollary, not an estimate transferred without proof to a nonradial path.

The first coefficient is rigorous as well.  Put
\(b_\sigma=C_{\rm nf}/(2a_\sigma)\), \(h_{0,\sigma}=1+b_\sigma/t\), and
\(L_\sigma h=h''+2a_\sigma h'-gh\).  Exact collection gives
\begin{equation}
 L_\sigma h_{0,\sigma}=
 \frac{b_\sigma(C_{\rm nf}+2)}{t^3}+\frac{4\mathcal{H}}{t^4}
 +\frac{4Hb_\sigma}{t^5}-\frac{B_1^2}{t^6}
 -\frac{B_1^2b_\sigma}{t^7}.
 \label{eq:olver-one-term-residual}
\end{equation}
If
\(J_{\sigma,d}=|a_\sigma|^{-1}\int_{\Gamma_d(t)}
|L_\sigma h_{0,\sigma}(u)|\,|du|\), the same proof yields
\begin{equation}
 w_{\sigma,d}=e^{a_\sigma t}
 \left(1+\frac{C_{\rm nf}}{2a_\sigma t}+\rho_{1,\sigma,d}\right),\qquad
 |\rho_{1,\sigma,d}(t)|\leq J_{\sigma,d}(t)e^{G_d(t)}.
 \label{eq:olver-one-term-bound}
\end{equation}
Indeed \(L_\sigma\rho_{1,\sigma,d}=-L_\sigma h_{0,\sigma}\), so the
inhomogeneous Volterra equation contains \(-L_\sigma h_{0,\sigma}\); its
sign disappears only after taking absolute values.  On a translated path,
the integral entering \(|a_\sigma|J_{\sigma,d}\) is at most
\begin{equation}
 \frac{|b_\sigma(C_{\rm nf}+2)|}{2x_d^2}+\frac{4|\mathcal{H}|}{3x_d^3}
 +\frac{|Hb_\sigma|}{x_d^4}+\frac{|B_1|^2}{5x_d^5}
 +\frac{|B_1|^2|b_\sigma|}{6x_d^6}.
 \label{eq:olver-one-term-majorant}
\end{equation}
Thus \(\rho_{1,\sigma,d}\) is holomorphic and \(O(x_d^{-2})\) uniformly on
each half-plane, hence \(O(t^{-2})\) uniformly on every covered closed proper
subsector.  Inserting
\(1+b_\sigma/t+O(t^{-2})\) in the exact equation gives the same
\(b_\sigma=C_{\rm nf}/(2a_\sigma)=u_1^{(\sigma)}\) as the direct recurrence.
Since \(e^{B_1/(2t^2)}=1+O(t^{-2})\), undoing
\eqref{eq:olver-exact-substitution} gives, with
\(\sqrt{-qz}=st\) and \(z^{1/4-B_2/2}=t^{1/2-B_2}\),
\begin{equation}
 \Psi_\sigma(z)=e^{2\sigma\sqrt{-qz}}z^{1/4-B_2/2}
 \left[1+u_1^{(\sigma)}z^{-1/2}+O(z^{-1})\right],\qquad
 u_1^{(\sigma)}=\frac{C_{\rm nf}}{2a_\sigma},
 \label{eq:olver-z-result}
\end{equation}
uniformly on the corresponding covering-sector images.  In particular the
relative error in the leading approximation is uniformly \(O(|t|^{-1})\),
that is, \(O(|z|^{-1/2})\).

Write \(q=|q|e^{i\phi}\), \(z=re^{i\theta}\), with angles modulo
\(2\pi\).  The phases \(S_\sigma=2\sigma\sqrt{-qz}\) have equal moduli when
\(\Re\sqrt{-qz}=0\), giving \(\theta=-\phi\pmod{2\pi}\); we call this the
Stokes ray.  They have maximal exponential separation when
\(\Im\sqrt{-qz}=0\), giving \(\theta=\pi-\phi\pmod{2\pi}\); we call this the
anti-Stokes ray.  Some literature reverses the names.  On the cover the same
geometry is \(\Re(st)=0\) and \(\Im(st)=0\).  Each line has two opposite
\(t\)-rays which \(z=t^2\) identifies as one \(z\)-ray.  Locating these rays
does not compute their Stokes multipliers.

Figure~\ref{fig:stokes-geometry} displays this covering geometry and its
rotation with the drive phase.  More explicitly, our operative convention is
that \(\Re\Delta S=0\) denotes equal-modulus (Stokes) curves and
\(\Im\Delta S=0\) denotes constant-phase (anti-Stokes) curves, with
\(\Delta S=4st\).  The names are not used without these equations.  Although
the rays rotate in fixed Bargmann coordinates when \(\arg F=\arg q\) changes,
all drive phases are physically equivalent: for \(U_\chi=e^{i\chi N}\),
\(U_\chi H(F)U_\chi^\dagger=H(Fe^{i\chi})\).  Hence the spectrum depends on
\(|F|\), not on \(\arg F\).  Moreover, \(|q|\) changes only the magnitude of
\(\Re\Delta S\), and therefore the dominance contrast, not the ray angles.
The exponentials are single-valued on the \(t\)-cover, but
\(S_+(-t)=S_-(t)\) and \(S_-(-t)=S_+(t)\).  Hence \(t\) and \(-t\) lie above
the same physical \(z\) while exchanging the exponential labels; dominance
in the ordinary \(z\)-plane is meaningful only after a sheet/branch choice.

\begin{figure}[tbp]
 \centering
 \includegraphics[width=0.96\textwidth]{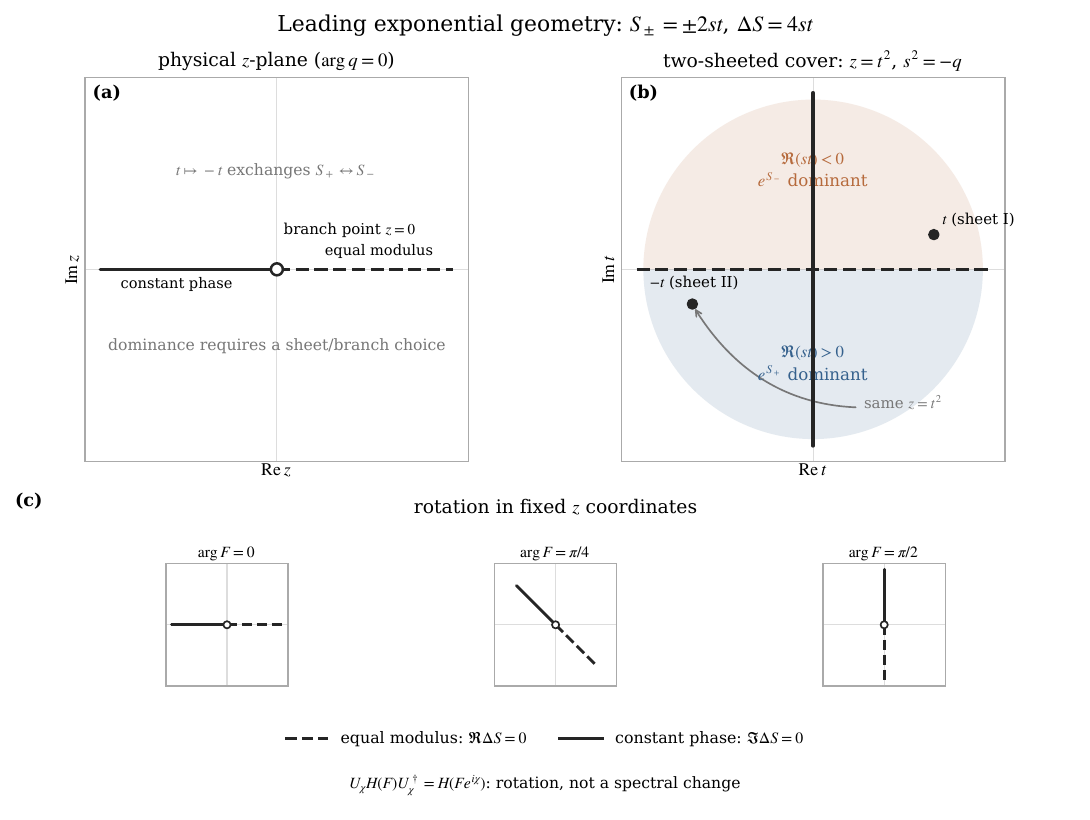}
 \caption{Stokes geometry of the leading factors
 \(S_\pm=\pm2\sqrt{-qz}=\pm2st\), with \(z=t^2\) and \(s^2=-q\).
 Our convention is: dashed curves are equal-modulus curves
 \(\Re\Delta S=0\) (Stokes), while solid curves are constant-phase curves
 \(\Im\Delta S=0\) (anti-Stokes), where \(\Delta S=4st\).  The exponentials
 are single-valued on the \(t\)-cover.  Opposite points \(t\) and \(-t\) lie
 above the same \(z\) but exchange \(S_+\leftrightarrow S_-\); consequently,
 dominance labels in the physical \(z\)-plane require a sheet/branch choice.
 On the cover, blue and orange indicate dominance of \(e^{S_+}\) and
 \(e^{S_-}\), respectively.  Changing \(\arg F=\arg q\) rotates the geometry in fixed
 Bargmann coordinates but, because
 \(U_\chi H(F)U_\chi^\dagger=H(Fe^{i\chi})\), does not change the spectrum.
 The angular geometry is independent of \(|q|\), which controls only the
 strength of the exponential contrast.}
 \label{fig:stokes-geometry}
\end{figure}

This is the fixed-parameter integral-equation error-control method from
Olver's chapter on asymptotic solutions of linear differential equations
\cite{Olver1997Asymptotics}; no large parameter is introduced.  We have
proved actual uniquely normalized solutions, uniform bounds, coefficient
agreement, and leading phase geometry.  This sectorial result provides the
outer canonical solutions for a future origin-to-infinity connection; Stokes
multipliers, the Bargmann determinant, and eigenvalue asymptotics belong to
that global step.  ``Dominant'' means larger relative exponential size.
Because both branches grow as \(\exp(O(|z|^{1/2}))\), the Gaussian Bargmann
norm alone does not justify discarding either one.

The covering equation has no turning point near infinity because its leading
coefficient \(-4q\) is nonzero; finite turning points nevertheless control the
strong-drive connection problem treated next.

This section derives a formal series expansion of the Bargmann solutions at
infinity in the covering variable, and shows that the corresponding
leading-order Liouville--Green approximation does not reproduce the full
first coefficient, requiring an explicit transport correction.  Following
Olver's error-bound method, it then proves the existence of uniquely
normalized, holomorphic canonical solutions on suitable half-planes, with
explicit bounds on their remainders, and displays the resulting Stokes
geometry.  These sectorial solutions supply outer building blocks for a
future origin-to-infinity connection; they do not by themselves determine
Stokes multipliers, the Bargmann determinant, or eigenvalue asymptotics.

\section{Strong drive: coalescing turning points and the Weber connection problem}
\label{sec:strong-drive}
\label{subsec:strong-drive-olver}

Let
\begin{equation}
 \eta=\frac{|F|}{V},\qquad \delta=\frac{\hbar\Delta}{V},
 \qquad \eta\longrightarrow\infty,
 \label{eq:strong-drive-parameters}
\end{equation}
with \(\delta\) fixed.  Unitary phase equivalence permits the choice
\(F=|F|>0\); hence
\begin{equation}
 B_1=q=2\eta,\qquad B_2=2\delta,\qquad B_3=-2E/V.
 \label{eq:strong-drive-WI-map}
\end{equation}
Substitution in the definitions of \(D,\mathcal{H},C_{\rm nf}\), independently of any scaling,
gives
\begin{equation}
 D=-2E/V+\delta-\delta^2,\quad \mathcal{H}=2\eta(1-\delta),\quad
 C_{\rm nf}=-8E/V+4\delta-4\delta^2-\frac34,
\end{equation}
and the exact covering equation is
\begin{equation}
 w''(t)=\left[-8\eta-\frac{C_{\rm nf}}{t^2}
 -\frac{8\eta(1-\delta)}{t^4}+\frac{4\eta^2}{t^6}\right]w(t).
 \label{eq:strong-drive-exact-cover}
\end{equation}

Balancing the constant and inverse-sixth-power terms gives
\(t=\eta^{1/6}x\); balancing the energy part of \(C_{\rm nf}/t^2\) with them then
gives \(E/V=\eta^{4/3}e\).  With \(u=\eta^{2/3}\), direct substitution yields
\begin{equation}
 \frac{d^2w}{dx^2}=\{u^2f(x;e)+u g(x)+h(x)\}w,
 \label{eq:strong-drive-scaled}
\end{equation}
where
\begin{equation}
 f=\frac4{x^6}+\frac{8e}{x^2}-8,\qquad
 g=-\frac{8(1-\delta)}{x^4},\qquad
 h=\frac{4\delta^2-4\delta+3/4}{x^2}.
 \label{eq:strong-drive-fgh}
\end{equation}
Thus the natural expansion parameter after extraction of the leading energy
scale is \(u^{-1}=\eta^{-2/3}\), not \(\eta^{-1}\) or a small-drive Taylor
parameter.  In particular, \(4\eta^2/t^6=4\eta/x^6\) before the derivative
rescaling and contributes \(4u^2/x^6\) in
\eqref{eq:strong-drive-scaled}: it is leading order.

Putting \(y=x^2\), the turning points satisfy
\begin{equation}
 1+2ey^2-2y^3=0.
\end{equation}
The exact derivative is
\(P_y=4ey-6y^2=2y(2e-3y)\).  Thus a nonzero double root has
\(e=3y/2\), and substitution in \(P=0\) gives \(1+y^3=0\).  The real
branch is
\begin{equation}
 e_0=-\frac32,\qquad
 1-3y^2-2y^3=-(y+1)^2(2y-1).
 \label{eq:strong-drive-factorization}
\end{equation}
The local strong-drive geometry selects the coalescing points \(x=i\) and
\(x=-i\) on sheets related by \(t\mapsto-t\); both cover the same negative
real Bargmann point \(z=t^2\).  The negative physical displacement caused by
a positive drive and the sheet-exchange rule
\eqref{eq:formal-sheet-exchange} therefore suggest this pair as physically
relevant.  Identifying it with the actual global continuation contour would,
however, require the still-missing construction of progressive paths from the
compensated branch \eqref{eq:compensated-branch} through the turning region
and the corresponding global connection map.  A positive-real-\(x\)
Schr\"odinger analogy alone neither supplies that construction nor captures
the required sheet exchange.

The complete \(u\)-dependent turning-point audit is obtained from
\begin{equation}
 Q_u=f+u^{-1}g+u^{-2}h,\qquad
 P_u(y)=1-\frac{2(1-\delta)}u y+
 \left(2e+\frac{4\delta^2-4\delta+3/4}{4u^2}\right)y^2-2y^3.
 \label{eq:strong-drive-full-Q}
\end{equation}
For \(e=-3/2+e_1/u+O(u^{-2})\), its two roots near \(-1\) and the
remaining root are
\begin{align}
 y_\pm={}&-1\pm v u^{-1/2}
 +\frac{-\delta+4e_1+1}{9u}+O(u^{-3/2}),\qquad
 v^2=\frac23(\delta-1-e_1),\nonumber\\
 y_s={}&\frac12+\frac{2\delta+e_1-2}{9u}+O(u^{-2}).
 \label{eq:strong-drive-all-turning-points}
\end{align}
Their six square roots comprise two points near \(i\), their negatives near
\(-i\), and the simple pair near \(\pm1/\sqrt2\).  In the \(t\)-plane they
are multiplied by \(u^{1/4}\), while in the \(z=t^2\) plane the opposite
square roots project together to \(z=u^{1/2}y\).  The pole of order six at
\(x=0\) remains excluded.  In particular, the simple pair cannot be omitted
from a global-domain claim.

Two coalescing simple turning points require a parabolic-cylinder (Weber)
comparison equation, as in Olver's two-turning-point construction
\cite[Chapter~12, Sections~1--4]{Olver1997Asymptotics}; a Bessel comparison
has the wrong turning-point class.

\begin{figure}[htbp]
\centering
\includegraphics[width=\textwidth]{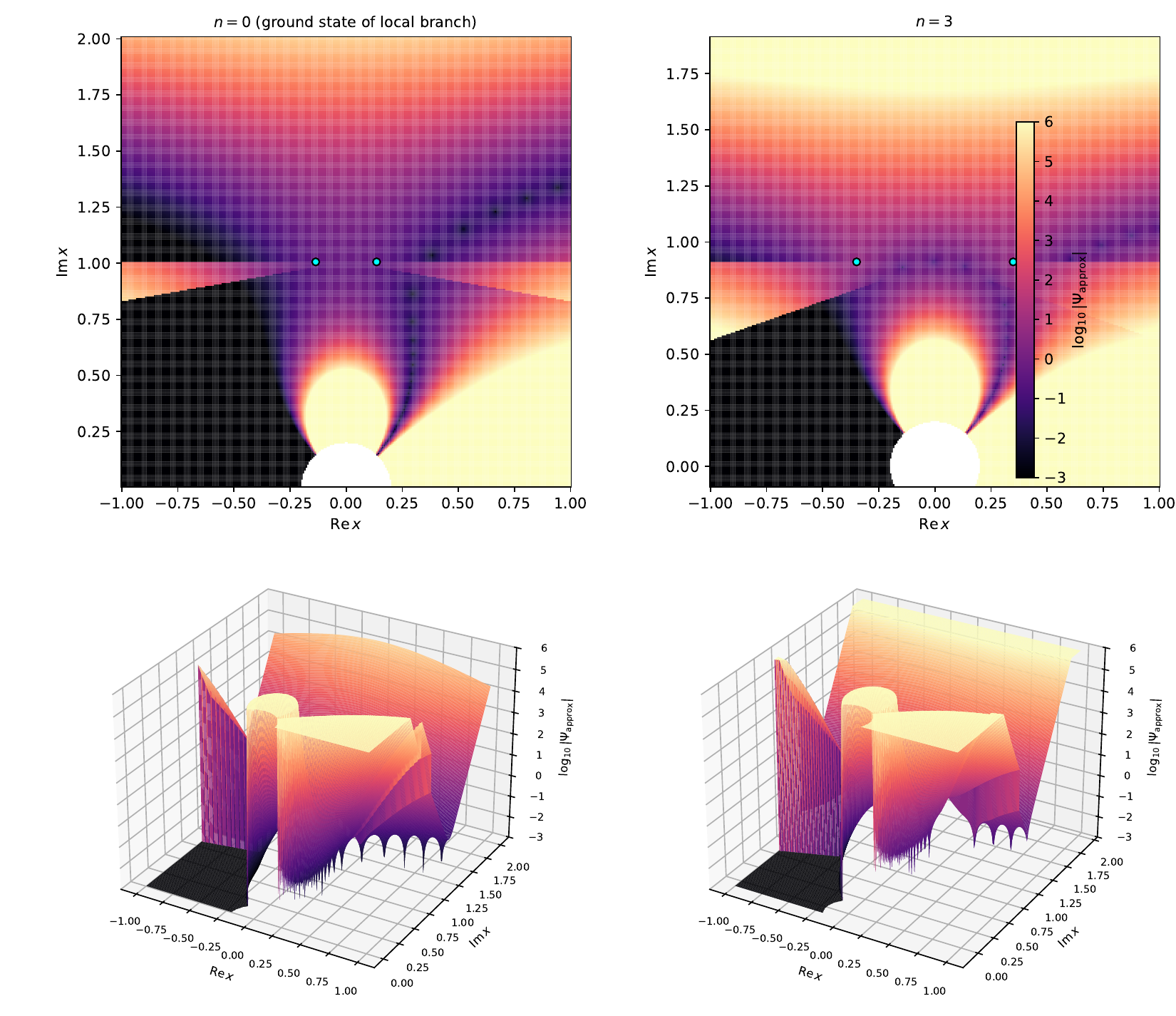}
\caption{Leading-order modulus of the Weber-approximate solution,
\(|\Psi_{\rm approx}(x)|=|(\mathrm d\zeta/\mathrm dx)^{-1/2}
D_p(\sqrt{2u}\,\zeta(x))|\), shown as a contour map (top) and as a surface
with \(\log_{10}|\Psi_{\rm approx}|\) on the vertical axis (bottom), near
the coalescing turning-point pair, for \(2|F|/V\approx45.25\) and zero
detuning \(2\hbar\Delta/V=0\) (i.e.\
\(u=8\), \(\delta=0\)).  The two columns compare the ground state of the
local branch (\(n=0\), \(2E_0/V\approx-194.1\)) with \(n=3\)
(\(2E_3/V\approx-111.0\)), using the boxed formula
\eqref{eq:strong-drive-energy}; the second coefficient of that formula is
conditional on the unfinished global error control noted there, so the
energies used here are illustrative rather than certified eigenvalues.
Markers are the exact turning points of the coalescing pair used in the
local construction, indicating where the leading local coefficient
\(u^2(\zeta^2-a_u^2)\) vanishes, separating the oscillatory sector from
the exponentially growing/decaying ones.  The faint fringes visible along
\(\arg(\sqrt{2u}\,\zeta)\approx\mp3\pi/4\) trace the anti-Stokes lines of
\(D_p\), where its two asymptotic terms have comparable magnitude and
interfere; they are a genuine feature of the comparison function, not a
numerical artifact.  The figure is purely illustrative of the local
Stokes geometry; it does not itself establish any global connection
statement, and the leading-order (\(u^{-1}\to0\)) construction used here
omits the \(\psi_u\) correction of \eqref{eq:olver-weber-remainder}.}
\label{fig:weber-modulus}
\end{figure}

To absorb the complete order-\(u\)
coefficient, take a simply connected local cut domain \(\mathcal X\)
containing the selected pair, but excluding \(x=0\) and the other four
zeros, and fix the branch of \(\zeta\) by
\begin{equation}
 \int_{\alpha_u}^{x}\sqrt{Q_u(s)}\,ds
 =\int_{-a_u}^{\zeta}\sqrt{v^2-a_u^2}\,dv,
 \qquad
 \frac{\pi i}{2}a_u^2=\int_{\alpha_u}^{\beta_u}\sqrt{Q_u(s)}\,ds,
 \label{eq:olver-weber-map}
\end{equation}
with compatible branches, \(\zeta(\alpha_u)=-a_u\), and
\(\zeta(\beta_u)=a_u\).  Reversing an orientation or a square-root branch
changes the corresponding action sign.  The Liouville transformation
\(W=(d\zeta/dx)^{1/2}w\) gives exactly
\begin{equation}
 W_{\zeta\zeta}=\{u^2(\zeta^2-a_u^2)+\psi_u(\zeta)\}W,
 \label{eq:olver-weber-transformed}
\end{equation}
where
\begin{equation}
 \psi_u=\left(\frac{dx}{d\zeta}\right)^{1/2}
 \frac{d^2}{d\zeta^2}\left(\frac{dx}{d\zeta}\right)^{-1/2}.
 \label{eq:olver-weber-remainder}
\end{equation}
Indeed, writing \(p=d\zeta/dx\) gives \(w=p^{-1/2}W\) and
\(d/dx=p\,d/d\zeta\).  Differentiating twice cancels the
\(W_\zeta\) term and leaves the Liouville contribution
\begin{equation}
 p^{-1/2}\frac{d^2}{d\zeta^2}p^{1/2}
 =-\frac34p^{-4}p_x^2+\frac12p^{-3}p_{xx},
 \label{eq:olver-liouville-identity}
\end{equation}
which is the second term in \eqref{eq:olver-weber-remainder} because
\(dx/d\zeta=p^{-1}\).
Level curves of the real part of either action in
\eqref{eq:olver-weber-map} map the original Stokes curves to the standard
Weber graph.  Canonical sectorial solutions are the appropriate rotations of
\(D_p(\sqrt{2u}\,\zeta)\), selected subdominant in the corresponding standard
Weber sectors and continued along the mapped progressive paths.

A complete global uniform error theorem is not claimed here.  Put
\(R_u(\zeta)=\psi_u(\zeta)\).  Variation of constants with a canonical Weber pair
\(U_1,U_2\), Wronskian \(\mathcal W\), gives the exact equation
\begin{equation}
 \epsilon(\zeta)=\frac1{\mathcal W}
 \int_{\Gamma(\zeta)}[U_1(\zeta)U_2(v)-U_2(\zeta)U_1(v)]
 R_u(v)[U(v)+\epsilon(v)]\,dv.
 \label{eq:weber-error-integral}
\end{equation}
If the kernel divided by the selected Weber modulus/envelope is bounded by
\(K_u(\zeta,v)\) along the progressive path, successive approximation gives
the concrete absolute bound
\begin{equation}
 |\epsilon(\zeta)|\leq M_U(\zeta)
 \left\{\exp\!\left(\int_{\Gamma(\zeta)}
 K_u(\zeta,v)|R_u(v)|\,|dv|\right)-1\right\}.
 \label{eq:weber-conditional-bound}
\end{equation}
It remains meaningful across zeros because it is absolute, not relative.
The construction is conditional on uniform finiteness of the displayed
variation.  Its global completion requires one domain reaching both the
irregular origin and the outer matching sectors; the displayed estimate is
the error-control criterion for that connection problem.

The endpoint expansion makes the obstruction explicit.  In a fixed
punctured origin sector,
\begin{equation}
 \sqrt{Q_u(x)}=\frac2{x^3}-\frac{2(1-\delta)}{ux}+O(x),\qquad
 \int^x\sqrt{Q_u(s)}\,ds=-\frac1{x^2}
 -\frac{2(1-\delta)}u\log x+O(1),
 \label{eq:strong-drive-origin-action}
\end{equation}
up to the simultaneous branch-dependent sign and an additive constant.
The Weber action is
\(\zeta^2/2-(a_u^2/2)\log\zeta+O(1)\), so
\(\zeta\sim\pm i\sqrt2/x\), sector by sector.  At infinity,
\(\int^x\sqrt{Q_u}\,dx=\sqrt{-8}\,x+O(x^{-1})\), whence
\(\zeta^2/2\sim\sqrt{-8}\,x\).  Thus both \(x=0\) and \(x=\infty\) lead to
different sectors of \(\zeta=\infty\), while the additional simple turning
points remain relevant to possible continuations.  No univalent single
Weber domain with all asserted endpoints has been proved.

For reference, \(Z=\sqrt{2u}\,\zeta\) gives
\begin{equation}
 W_{ZZ}=\left(\frac{Z^2}{4}-p_u-\frac12\right)W,\qquad
 p_u+\frac12=\frac{u a_u^2}{2}.
 \label{eq:strong-drive-weber-index}
\end{equation}
The exact parabolic-cylinder formula DLMF~12.2.19 \cite{DLMF12}, together with
\(D_p(Z)=U(-p-1/2,Z)\) (DLMF~12.2.5), contains a model coefficient
proportional to \(1/\Gamma(-p)\).  Its zeros cannot yet be identified with
physical zeros because the required rotations, origin branches, gauge
monodromy, and progressive domain chain have not been constructed.  The
first missing hypothesis is therefore the comparison-domain hypothesis,
before any assertion that the Volterra variation is small.  A minimal repair
is a finite Weber--Airy--endpoint domain-chain lemma proving univalence,
progressive paths, sheet exchange, and \(o(1)\) transition errors.

The relation to the physical wavefunction must be made before naming either
comparison solution.  Combining \eqref{eq:olver-exact-substitution},
\(t=\eta^{1/6}x\), and \(W=(d\zeta/dx)^{1/2}w\) gives the complete pullback
\begin{equation}
 \boxed{\mathcal T_uW:=\Psi(z)=
 t^{1/2-B_2}\exp\!\left(\frac{B_1}{2t^2}\right)
 \left(\frac{d\zeta}{dx}\right)^{-1/2}W(\zeta(x)),}
 \quad z=t^2,\quad B_1=2\eta,\quad B_2=2\delta.
 \label{eq:complete-weber-pullback}
\end{equation}
Here \(B_1/(2t^2)=u/x^2\).  On a consistent endpoint branch,
\eqref{eq:strong-drive-origin-action} and
\(\zeta\sim\pm i\sqrt2/x\) permit neutral canonical labels such that
\begin{equation}
 W_1\sim x^{-3/2+2\delta}e^{-u/x^2},\qquad
 W_2\sim x^{5/2-2\delta}e^{u/x^2},\qquad
 \left(\frac{d\zeta}{dx}\right)^{-1/2}\sim c\,x .
 \label{eq:neutral-weber-origin-pair}
\end{equation}
Since \(t^{1/2-2\delta}\) differs from \(x^{1/2-2\delta}\) only by a
nonzero \(u\)-dependent constant, direct multiplication gives
\begin{equation}
 \mathcal T_uW_1\sim1,\qquad
 \mathcal T_uW_2\sim x^{4-4\delta}e^{2u/x^2}
 \sim z^{2-2\delta}e^{2\eta/z}.
 \label{eq:weber-cancellation-enhancement}
\end{equation}
This is the covering-variable version of
\eqref{eq:compensated-branch}--\eqref{eq:enhanced-branch}, under the
identification \(B_1=2\eta\), \(B_2=2\delta\) already recorded in
\eqref{eq:complete-weber-pullback}.  Reversing the action or
square-root branch exchanges \(W_1,W_2\), but not the invariant distinction
between cancellation and enhancement of the gauge exponential.  Thus a
generic continued solution has the formal endpoint form
\begin{equation}
 \Psi(z;E)\sim C_{\rm reg}(E)[1+\cdots]+C_{\rm ess}(E)
 z^{2-2\delta}e^{2\eta/z}[1+\cdots].
 \label{eq:essential-endpoint-coefficient}
\end{equation}
Cancellation is not automatic, and Stokes continuation may restore a
\(W_2\) component even if it is absent initially.  These are formal endpoint
relations, not a proof of sectorial existence or holomorphic patching.

\begin{conjecture}[Global compensated-to-essential connection coefficient]
\label{conj:essential-connection}
There exists a finite canonical domain chain \(\Gamma_{\rm phys}\), with the
required sheet exchange, and canonically normalized bases
\((W_1^{(0)},W_2^{(0)})\) and \((W_1^{(m)},W_2^{(m)})\) in its initial and
terminal origin sectors, satisfying
\[
 \mathcal T_uW_1^{(j)}\sim1,\qquad
 \mathcal T_uW_2^{(j)}\sim z^{2-2\delta}e^{2\eta/z}.
\]
Its continuation operator obeys
\[
 \mathcal A_{\Gamma_{\rm phys}}W_1^{(0)}
 =C_{\rm reg}W_1^{(m)}+C_{\rm ess}W_2^{(m)},
\]
where, with the Wronskian order shown,
\begin{equation}
 \boxed{C_{\rm ess}(E;u,\delta)=
 \frac{\mathscr W_\zeta(
 \mathcal A_{\Gamma_{\rm phys}}W_1^{(0)},W_1^{(m)})}
 {\mathscr W_\zeta(W_2^{(m)},W_1^{(m)})}.}
 \label{eq:essential-wronskian}
\end{equation}
Bilinearity gives the numerator as \(C_{\rm ess}\) times the denominator,
fixing the sign.  Nonzero renormalizations multiply the coefficient by a
nonzero factor and leave its zero set unchanged.

Physical conjugation and sheet symmetries reduce the cancellation conditions
to \(C_{\rm ess}=0\), equivalent to a single-valued holomorphic germ at zero.
On the physical slice its zeros, with multiplicities, coincide with those of
\(\Xi\).  Equivalently, locally in \(E\) and away from the poles of \(\Xi\),
\(C_{\rm ess}=G\,\Xi\) for a holomorphic nonvanishing \(G\).  Near a fixed
level the expected Weber form is
\[
 C_{\rm ess}=\frac{A(E;u,\delta)}{\Gamma[-p_u(E)]}
 +R(E;u,\delta),\qquad A\ne0,\qquad R=O(u^{-1}A).
\]
The domain chain, symmetry reduction, determinant equivalence, and remainder
estimate in this statement are all conjectural.
\end{conjecture}

The independent formal operator prediction associated with this conjectural
zero localization is \eqref{eq:strong-drive-energy-constant}, with the
explicit coefficient \(K_n(\delta)\) in \eqref{eq:strong-drive-K}.  Its
operator derivation does not prove any assertion in the conjecture.

The modified-Bessel equation remains useful only outside this region.  For
\begin{equation}
 w_0''=(a^2-C_{\rm nf}/t^2)w_0,\qquad a^2=-8\eta,
\end{equation}
its solutions are \(\sqrt t I_\nu(at)\) and \(\sqrt t K_\nu(at)\), whose
Wronskian is \(1\) up to orientation, with
\begin{equation}
 \nu^2=\frac14-C_{\rm nf}=8E/V+(2\delta-1)^2.
 \label{eq:bessel-order}
\end{equation}
Direct differentiation checks both statements.  In the strong-drive scaling,
\(at\) and \(\nu\) are of order \(u\), generally with complex argument and
order.  Since the omitted inverse-sixth-power term is leading near the pair,
this is only an outer comparison.  Exact-to-Bessel Volterra error and
large-order/large-argument Bessel truncation error must be estimated
separately; no combined estimate is asserted.

Quantization must not discard an infinity branch: both have
\(\exp(O(|z|^{1/2}))\) growth.  Before the complete pullback, use neutral
sectorial labels \(W_1^{(j)},W_2^{(j)}\).  On overlaps one may formally write
\begin{equation}
 W_1^{(j+1)}=W_1^{(j)}+s_j(E)W_2^{(j)}.
 \label{eq:origin-stokes-multiplier}
\end{equation}
Conditional on constructing canonical sectorial solutions with these
asymptotics on sectors covering a punctured neighbourhood of the irregular
origin, the \(W_1\) branches selected by the pullback would patch to a
holomorphic germ exactly when
all independent inadmissible coefficients \(s_j(E)\) vanish and the algebraic
monodromy cancels as in \eqref{eq:compensated-branch}.  At present the
\(s_j\) are formal sectorial connection data: their required global
construction and evaluation have not been supplied.  Before physical
conjugation and sheet symmetry they constitute several conditions; a proved
reduction to one physical scalar Olver condition is also absent.  The exact
scalar condition available independently from the earlier entireness theorem
is
\begin{equation}
 \boxed{\Xi(E)=0.}
 \label{eq:strong-drive-exact-quantization}
\end{equation}
This is not a newly derived Olver connection formula.  Its zeros are exactly
the Bargmann eigenvalues by
Theorem~\ref{thm:bargmann-characteristic}.  A nonzero rescaling of a sectorial
solution only rescales its Wronskians and connection coefficients, leaving
their common zero set unchanged.  The individual \(s_j\) remain formal;
Conjecture~\ref{conj:essential-connection} gives a precise proposed global
coefficient, not its construction or evaluation.

The local Weber reduction determines the first two fixed-level coefficients,
the second conditionally on identifying the physical global connection zero
with the canonical Weber zero.  Write \(e=e_0+e_1/u+O(u^{-2})\).  At
\(x=i,e_0=-3/2\),
\begin{equation}
 \tfrac12f_{xx}=48,\qquad f_e=-8,\qquad g=-8(1-\delta).
\end{equation}
For \(X=48^{1/4}u^{1/2}(x-i)\), the leading local equation is
\begin{equation}
 w_{XX}=\{X^2+\lambda\}w,\qquad
 \lambda=-\frac{2}{\sqrt3}(e_1+1-\delta).
\end{equation}
The canonical solution \(D_p(\sqrt2X)\) has \(\lambda=-2p-1\).  Its
fixed-level connection zero \(p=n\) therefore gives
\begin{equation}
 \boxed{\frac{E_n}{V}=-\frac32\eta^{4/3}
 +\left[\delta+\sqrt3\left(n+\frac12\right)-1\right]\eta^{2/3}
 +O(1),\qquad n\ \hbox{fixed}.}
 \label{eq:strong-drive-energy}
\end{equation}
The leading coefficient follows rigorously from the double-root geometry;
the second remains conditional on the unfinished global error-control and
connection step.  No uniformity for growing \(n\) is claimed.

\subsection{Independent Bogoliubov operator check}

As an independent operator check, put \(a=b-r\), where
\(r>0\) satisfies \(r^3+\delta r=\eta\).  Exact normal-ordered expansion,
including cancellation of the linear term, gives
\begin{align}
 \frac HV={}&\left(\frac12r^4+\delta r^2-2\eta r\right)
 +(2r^2+\delta)b^\dagger b
 +\frac{r^2}{2}(b^{\dagger2}+b^2)\nonumber\\
 &-r(b^{\dagger2}b+b^\dagger b^2)+\frac12b^{\dagger2}b^2.
 \label{eq:shifted-hamiltonian}
\end{align}
Using \(\eta=r^3+\delta r\), the constant is exactly
\(-3r^4/2-\delta r^2\).
Thus the exact quadratic coefficient is \(A=2r^2+\delta\).  Bogoliubov
diagonalization gives
\(\Omega=\sqrt{(2r^2+\delta)^2-r^4}=\sqrt3r^2+O(1)\) and quadratic energy
\begin{equation}
 \frac{E_n^{(2)}}V=\frac12r^4+\delta r^2-2\eta r
 +\Omega\left(n+\frac12\right)-\frac{2r^2+\delta}{2}.
\end{equation}
Choose the real transformation \(b=Cc+Sc^\dagger\), where
\(\tanh(2\theta)=-r^2/(2r^2+\delta)\), \(C=\cosh\theta\), and
\(S=\sinh\theta\).  At leading order
\[
 C^2+S^2=\frac2{\sqrt3},\qquad 2CS=-\frac1{\sqrt3}.
\]
Parity makes the first-order cubic correction vanish.  Exact ladder-operator
algebra in the \(c\)-number basis gives the first-order quartic and
second-order cubic terms
\begin{align}
 \Delta_{4,n}^{(1)}&=
 \frac{18n^2+(18-16\sqrt3)n+15-8\sqrt3}{24},\nonumber\\
 \Delta_{3,n}^{(2)}&=-
 \frac{30n^2+(30-24\sqrt3)n+23-12\sqrt3}{36},
 \label{eq:cubic-quartic-constant}
\end{align}
and hence
\begin{equation}
 \Delta_{4,n}^{(1)}+\Delta_{3,n}^{(2)}
 =-\frac{6n^2+6n+1}{72}.
\end{equation}
The expansion
\[
 r=\eta^{1/3}-\frac{\delta}{3}\eta^{-1/3}
 +\frac{\delta^3}{81}\eta^{-5/3}+O(\eta^{-7/3})
\]
in the exact quadratic energy contributes
\(\delta(1-2\delta)/6\) at order one; its other \(n\)-dependent constant
terms cancel internally.  Thus formal Rayleigh--Schr\"odinger perturbation
theory, independently of the global connection conjecture, yields
\begin{equation}
 \boxed{\begin{aligned}
 \frac{E_n}{V}={}&-\frac32\eta^{4/3}
 +\left[\delta+\sqrt3\left(n+\frac12\right)-1\right]\eta^{2/3}\\
 &+\frac{\delta(1-2\delta)}6-\frac{6n^2+6n+1}{72}
 +K_n(\delta)\eta^{-2/3}+O(\eta^{-4/3})
 \end{aligned}}
 \label{eq:strong-drive-energy-constant}
\end{equation}
where the exactly evaluated coefficient is
\begin{align}
 K_n(\delta)=\frac1{2592}\bigl\{&
 96\delta^3+(144\sqrt3\,n-288+72\sqrt3)\delta^2\nonumber\\
 &+(144n^2-288\sqrt3\,n+144n-144\sqrt3+312)\delta\nonumber\\
 &+10\sqrt3\,n^3+(15\sqrt3-144)n^2
 +(139\sqrt3-144)n+67\sqrt3-120\bigr\}.
 \label{eq:strong-drive-K}
\end{align}
This is an algebraically verified coefficient of the formal fixed-\(n\)
Rayleigh--Schr\"odinger expansion, not a consequence of the unproved Olver
connection conjecture and not a remainder estimate uniform in \(n\).
It implies
\begin{align}
 \frac{E_{n+1}-E_n}{V}
 &=\sqrt3\,\eta^{2/3}-\frac{n+1}{6}
 +L_n(\delta)\eta^{-2/3}+O(\eta^{-4/3}),\nonumber\\
 \frac{E_{n+2}-2E_{n+1}+E_n}{V}
 &=-\frac16+M_n(\delta)\eta^{-2/3}+O(\eta^{-4/3}),
 \label{eq:strong-drive-spacing}
\end{align}
with
\begin{align}
 L_n(\delta)&=\frac1{1296}\bigl[
 72\sqrt3\,\delta^2+(144n-144\sqrt3+144)\delta\nonumber\\
 &\hspace{22mm}+15\sqrt3\,n^2+(30\sqrt3-144)n
 +82\sqrt3-144\bigr],\nonumber\\
 M_n(\delta)&=\frac{48\delta+10\sqrt3\,n+15\sqrt3-48}{432}.
\end{align}

For reproducibility, set \(\lambda=r^{-1}\) and expand
\(\lambda^2(H/V-E_{\rm vac}^{(2)})=\sum_{j=0}^4\lambda^jK_j+O(\lambda^5)\).
Intermediate-normalization recursion retains every reachable number state
and gives \(\kappa_1=\kappa_3=0\).  The independently evaluated finite-sum
organization includes the finite-\(r\) vertex and denominator corrections,
\(V_4^2\), all \(V_3^2V_4\) permutations, \(V_3^4\), and both normalization
subtractions.  The two calculations give identical \(\kappa_2,\kappa_4\);
only after their combination is \(r^3+\delta r=\eta\) re-expanded to obtain
\eqref{eq:strong-drive-K}.

For arbitrary phase, applying
\(U_\chi=e^{i\chi N}\), \(\chi=-\arg F\), transports the positive-drive
eigenfunction; the formal expansion \eqref{eq:strong-drive-energy-constant}
remains valid with
\(\eta=|F|/V\).

This section identifies the coalescing pair of turning points that
controls the strong-drive limit and reduces the local connection problem,
via the covering-variable normal form, to a parabolic-cylinder comparison
equation following Olver's two-turning-point construction; this is one of
two independent WKB routes to the strong-drive asymptotics developed in
this work.  It states as a conjecture the global compensated-to-essential
connection coefficient whose vanishing would give the exact quantization
condition, with the leading energy term rigorous from the double-root
geometry alone and the remaining terms conditional on that conjecture.  A
separate, non-WKB Bogoliubov operator expansion then reproduces the same
leading and next-order energy coefficients as an algebraically verified
formal Rayleigh--Schr\"odinger result, independently of the connection
conjecture, and supplies a further coefficient not reached by the
conjectural route.

\section{An independent complex-WKB route}
\label{sec:sectorial-solutions}

The strong-drive equation also admits a direct Fedoryuk normal form, without
passing through the preceding covering transformation.  With \(F=|F|>0\),
\(-B_3=2E/V=2\eta^{4/3}e\), the exact substitution
\begin{equation}
 \begin{aligned}
  \Psi(z)&=z^{-\delta}e^{\eta/z}u(z),\\
  z&=\eta^{1/3}y, & \varepsilon&=\eta^{-2/3}.
 \end{aligned}
 \label{eq:fedoryuk-direct-map}
\end{equation}
on a fixed punctured sector gives
\begin{equation}
 \begin{aligned}
  \varepsilon^2u_{yy}
   &=\left\{Q_0+\varepsilon Q_1+\varepsilon^2Q_2\right\}u,\\
  Q_0&=\frac{P(y;e)}{y^4}, &
  Q_1&=\frac{2(\delta-1)}{y^3}, &
  Q_2&=\frac{\delta(\delta-1)}{y^2}.
 \end{aligned}
 \label{eq:fedoryuk-normal-form}
\end{equation}
where \(P(y;e)=1+2ey^2-2y^3\).  Thus the natural small
parameter is \(\eta^{-2/3}\), rather than \(\eta^{-1/3}\), and the leading
coefficient is meromorphic with a fourth-order pole at the genuine irregular
point \(y=0\).

For generic \(e\), the three zeros of \(P\) are simple.  At
\((y,e)=(-1,-3/2)\) two coalesce to a double turning point, while \(y=1/2\)
remains simple.  A branch
\(p=\sqrt{Q_0}\) on a cut domain defines
\(S_j(y)=\int_{y_j}^yp(t)\,dt\).  We use Fedoryuk's convention that Stokes
curves satisfy \(\Re S_j=0\); the terminology is reversed in part of the
physics literature.  The exact coefficient satisfies the local analytic
prerequisites for canonical-domain WKB constructions on a cut domain that
excludes the pole and turning points and for which progressive paths are
separately verified \cite[Chapter~III, Sections~2 and~4]{Fedoryuk1993Asymptotic}.
The local Airy model is appropriate at a separated simple point, whereas the
critical pair requires the multiple or two-close-turning-point models
\cite[Chapter~IV, Sections~2, 4, and~6--7]{Fedoryuk1993Asymptotic}.  The
complete required collection of such domains and paths has not been proved.

\begin{figure}[t]
 \centering
 \includegraphics[width=.92\linewidth]{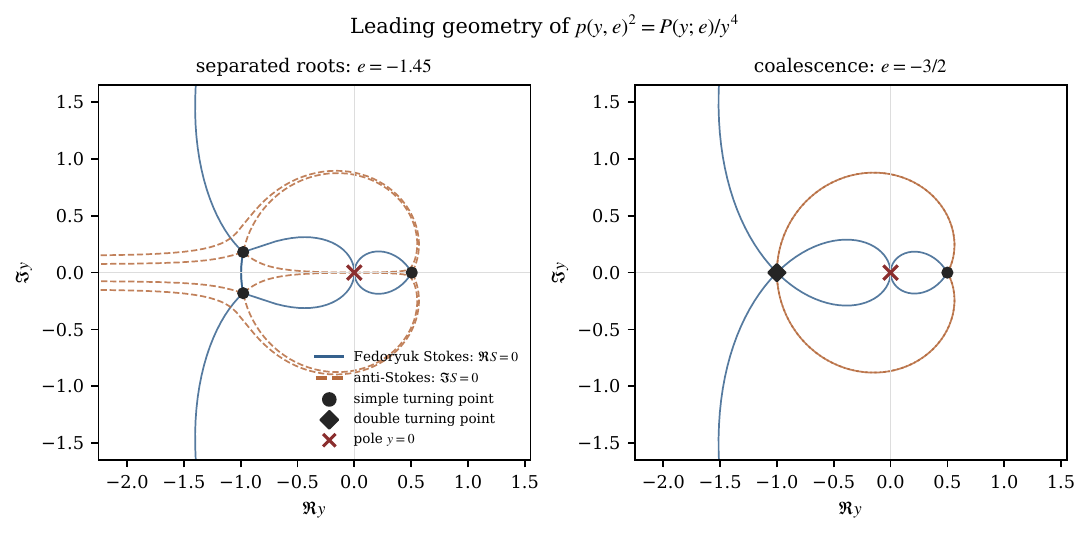}
 \caption{Leading quadratic-differential geometry for
 \(p^2=P(y;e)/y^4\), for a nearby separated configuration and at the real
 coalescence.  Turning-point locations, multiplicities, and initial incidence
 directions are analytic; the displayed continuations are numerically traced
 and do not prove global connectivity.  Solid curves use the Fedoryuk Stokes
 convention
 \(\Re S=0\), and dashed curves satisfy \(\Im S=0\).}
 \label{fig:fedoryuk-geometry}
\end{figure}

The close pair also has a direct local Weber reduction.  Set
\(e=-3/2+\varepsilon e_1+O(\varepsilon^2)\) and
\(y=-1+\varepsilon^{1/2}X\).  Expansion of the complete coefficient in
\eqref{eq:fedoryuk-normal-form}, including \(Q_1\), gives
\begin{equation}
 u_{XX}=\left\{3X^2+2(e_1+1-\delta)
 +O(\varepsilon^{1/2})\right\}u.
 \label{eq:fedoryuk-local-weber}
\end{equation}
With \(s=3^{1/4}X\), this is
\(u_{ss}=(s^2+\Lambda)u\), where
\(\Lambda=2(e_1+1-\delta)/\sqrt3\).  We use
\(D_\nu(Z)=U(-\nu-1/2,Z)\), so that
\(D_\nu(\sqrt2s)\) obeys
\(u_{ss}=\{s^2-(2\nu+1)\}u\).

There are two inequivalent opposite subdominant-sector pairs.  For the pair
centred on the real \(s\) axis, the Weber connection zero \(\nu=n\) gives
\begin{equation}
 e_{1,n}^{(R)}=\delta-1-\sqrt3(n+1/2),\qquad
 u_{n,R}^{(0)}=\phi_n(s),
 \label{eq:fedoryuk-real-branch}
\end{equation}
where
\(\phi_n(w)=\pi^{-1/4}(2^n n!)^{-1/2}e^{-w^2/2}H_n(w)\).
For the pair centred on the imaginary axis, put \(t=-is\).  The corresponding
zero instead gives
\begin{equation}
 e_{1,n}^{(I)}=\delta-1+\sqrt3(n+1/2),\qquad
 u_{n,I}^{(0)}=\phi_n(-is).
 \label{eq:fedoryuk-imaginary-branch}
\end{equation}
Indeed, the rotation reverses the sign in the index relation.  Since exactly
\(X=z+\eta^{1/3}\), their sectorial Bargmann pullbacks are
\begin{equation}
 \Psi_{n,\sigma}^{\rm loc}(z)=z^{-\delta}e^{\eta/z}
 \phi_n\!\left(c_\sigma3^{1/4}[z+\eta^{1/3}]\right),
 \quad c_R=1,\quad c_I=-i.
 \label{eq:fedoryuk-local-pullback}
\end{equation}
The logarithm in \(z^{-\delta}\) is fixed on the local sector; this expression
is centred near \(z=-\eta^{1/3}\), excludes the origin, and has only a local
normalization.  Neither the gauge factor nor the displayed product is thereby
proved to be a global Bargmann eigenfunction.

Thus the two local candidate energies are
\begin{equation}
 \frac{E_{n,\sigma}}{V}=-\frac32\eta^{4/3}
 +e_{1,n}^{(\sigma)}\eta^{2/3}+e_{2,n}^{(\sigma)}
 +O(\eta^{-2/3}),
 \qquad \sigma=R,I.
 \label{eq:fedoryuk-local-energies}
\end{equation}
Local intermediate-normalization solvability at the next order gives, for
both sector orientations,
\begin{equation}
 e_{2,n}^{(R)}=e_{2,n}^{(I)}
 =\frac{\delta(1-2\delta)}6-\frac{6n^2+6n+1}{72}.
 \label{eq:fedoryuk-local-constant}
\end{equation}
The corrected Stokes geometry does not select between these pairs.  We call
the assertion that the compensated Bargmann continuation selects the
imaginary-axis pair the \emph{Fedoryuk--Bargmann orientation hypothesis}.
Its positive-\(\sqrt3\) coefficient agrees with the independent
Bogoliubov result \eqref{eq:strong-drive-energy}, but this is only a
consistency check, not a proof of the hypothesis.  The first
\(O(\varepsilon^{1/2})\) local correction is obtained from an inhomogeneous
Weber equation; its odd cubic-plus-linear forcing couples only the Hermite
indices \(n\mathbin{\pm}1,n\mathbin{\pm}3\), and its diagonal solvability
term vanishes.  More precisely, if
\(A=10/3^{5/4}\) and
\(B_I=[2-2\delta+4\sqrt3(n+1/2)]/3^{3/4}\), then \(t=-is\) gives the
rotated differential perturbation
\begin{equation}
 P_I(t)=i(At^3-B_It),
 \qquad (\mathcal{L}_{\rm Herm}-2n-1)u_{1,I}=-P_I\phi_n.
 \label{eq:fedoryuk-rotated-perturbation}
\end{equation}
Thus its cubic and linear terms have opposite relative signs; they do not
acquire one common phase.  Exact finite Hermite sums solve both first-order
inhomogeneous equations, and substitution of \(u_0+\varepsilon^{1/2}u_1\)
leaves an absolute local residual \(O(\varepsilon)\), whereas \(u_0\) alone
has residual \(O(\varepsilon^{1/2})\).  For the rotated pair the first- and
second-order solvability conditions use the analytically continued contour
bilinear form, not a real-axis \(L^2\) inner product.  The equality of
\eqref{eq:fedoryuk-local-constant} with the independent operator constant in
\eqref{eq:strong-drive-energy-constant} is a consistency check only.

This yields a verified sectorial applicability framework.  Its global
completion is a finite chain connecting the gauge-compensated origin branch
to the sectors needed for a single-valued Bargmann germ, with control of the
close pair, any forced Airy passage, overlaps, and gauge/sheet monodromy.
Until that chain is established, a Fedoryuk spectral condition cannot be
identified with \(\Xi(E)=0\).  The local scaling
\(e+3/2=O(\varepsilon)\) independently explains energy corrections in steps
of \(\eta^{2/3}\), but assigning their \(n\)-dependent coefficients to
physical eigenvalues still requires the open global connection argument.

Set \(\mathcal A=(D_\nu(Z),D_{-\nu-1}(iZ))^T\) and
\(\mathcal B=(D_\nu(-Z),D_{-\nu-1}(-iZ))^T\), with principal arguments and
the indicated rotations.  Their members are subdominant about the
positive/negative real and negative/positive imaginary axes, respectively.
DLMF 12.2.17--12.2.20 \cite{DLMF12} and
\(D_\nu(Z)=U(-\nu-1/2,Z)\) give two independent identities:
\begin{equation}
 \begin{gathered}
 c_W=\frac{\sqrt{2\pi}e^{-i\pi(\nu+1)/2}}{\Gamma(-\nu)},\qquad
 d_W=\frac{\sqrt{2\pi}e^{i\pi\nu/2}}{\Gamma(\nu+1)},\\
 D_\nu(Z)=e^{-i\pi\nu}D_\nu(-Z)+c_WD_{-\nu-1}(iZ),\\
 D_{-\nu-1}(iZ)=d_WD_\nu(-Z)-e^{i\pi\nu}D_{-\nu-1}(-iZ),\\
 \mathcal A=M_W\mathcal B,\qquad
 M_W=\begin{pmatrix}e^{i\pi\nu}&-e^{i\pi\nu}c_W\\
 d_W&-e^{i\pi\nu}\end{pmatrix},\quad
 M_W^{-1}=M_W,\quad\det M_W=-1.
 \end{gathered}
 \label{eq:fedoryuk-weber-matrix}
\end{equation}
Indeed \(c_Wd_W=e^{i\pi\nu}-e^{-i\pi\nu}\), while the ordered
Wronskians are \(-ie^{-i\pi\nu/2}\) and \(ie^{-i\pi\nu/2}\).  Both
bases are nondegenerate for all complex \(\nu\), including \(\nu=n\),
where \(c_W(n)=0\) is simple and the matrix has a regular limit.  This exact
model matrix is not yet the physical transfer matrix.  In transported WKB
bases the latter has only the schematic factorization
\begin{equation}
 M_{\rm glob}=B_{\rm out}M_A^{\chi_A}B_2M_WB_1G_0,
 \qquad \chi_A\in\{0,1\},
 \label{eq:fedoryuk-global-matrix-scheme}
\end{equation}
where \(M_A\) occurs only if the continuation is forced past the remaining
simple point, the \(B_j\) are overlap and sheet-change matrices, and \(G_0\)
is the origin gauge transition.  Its entries cannot be assigned before the
paths and branches are constructed.  Absorb the exterior factors and
boundary data into \(\ell=(\ell_1,\ell_2)^T\) and
\(r=(r_1,r_2)^T\).  The candidate inadmissible coefficient is
\begin{equation}
 \ell^TM_Wr=e^{i\pi\nu}(\ell_1r_1-\ell_2r_2)
 -e^{i\pi\nu}c_W\ell_1r_2+d_W\ell_2r_1.
 \label{eq:fedoryuk-global-contraction}
\end{equation}
General exterior matrices therefore mix every Weber entry.  The choice
\(\ell=r=(1,0)^T\), for example, gives \(e^{i\pi\nu}\), so
\(c_W(n)=0\) is not a global zero; an unknown Airy transition can likewise
move or remove it.  Factorization through \(c_W\) requires
\(e^{i\pi\nu}(\ell_1r_1-\ell_2r_2)+d_W\ell_2r_1=0\) and
\(\ell_1r_2\ne0\).  Sectorial dominance alone does not prove these exact
alignment conditions.  On a positive circuit the factor
\(z^{-\delta}\) has multiplier \(e^{-2\pi i\delta}\); the transformed
solution must therefore acquire \(e^{2\pi i\delta}\) for their product to be
single-valued.

\begin{conjecture}[Global Fedoryuk--Bargmann connection hypotheses]
\label{conj:fedoryuk-bargmann-lemma}
H1 assumes only a uniformly finite chain of punctured-origin, WKB,
close-pair Weber and, if needed, Airy domains, with uniform overlaps,
progressive paths, controlled pole and comparison transitions, compatible
sheets and monodromy, and identification of the initial solution with the
gauge-compensated Bargmann germ.  It assumes no orientation, gamma zero, or
energy formula.  H2 assumes only that one analytic boundary contraction,
minor, or determinant, with controlled remainder, measures the final
inadmissible branch.  A separate H3 is required: the transported boundary
vectors preserve the factorizing condition following
\eqref{eq:fedoryuk-global-contraction}, to sufficient uniform order, with a
nonzero coefficient of \(c_W\).  Neither H1 nor H2 is known to imply H3.
\end{conjecture}

The verified local WKB and turning-point results isolate the remaining global
task: a univalent chain with progressive paths at the fourth-order pole,
uniform overlap errors, and Bargmann identification.  Under H1--H3, emergence
of the imaginary orientation from the constructed path, and independent
identification with \(\Xi(E)=0\), the simple gamma zero and a
Rouch\'e argument yield
\begin{equation}
 \begin{aligned}
 \mathcal C_{\rm Fed}(E,\eta,\delta)=0
 \quad\Longrightarrow\quad
 \frac{E_n}{V}={}&-\frac32\eta^{4/3}\\
 &+\left[\delta+\sqrt3\left(n+\frac12\right)-1\right]\eta^{2/3}\\
 &+\frac{\delta(1-2\delta)}6-\frac{6n^2+6n+1}{72}
 +O(\eta^{-2/3}).
 \end{aligned}
 \label{eq:fedoryuk-conditional-quantization}
\end{equation}
Thus the displayed energy is a consequence of this strengthened conditional
package, not of H1 or local Weber theory alone.
The positive-\(\sqrt3\) sign agrees independently with Bogoliubov and Olver,
and the constant agrees with formal Rayleigh--Schr\"odinger theory; neither
agreement proves the global connection.  In particular the known
\(K_n(\delta)\eta^{-2/3}\) term in
\eqref{eq:strong-drive-energy-constant} remains a formal operator
coefficient, not a Fedoryuk consequence.

This section develops the second independent WKB route to the strong-drive
asymptotics, working directly with the Fedoryuk normal form of the original
equation rather than through the covering-variable transformation of the
preceding section.  It identifies two candidate sectorial energy branches
and introduces, as a hypothesis, the Fedoryuk--Bargmann orientation
selecting one of them, noting that its agreement with the independently
computed Bogoliubov coefficients is a consistency check, not a proof.  A
further global-connection hypothesis is required before the resulting
quantization condition can be identified with the exact spectral condition
established elsewhere in the paper.

\section{Weak-drive expansion and edge resummation}
\label{sec:weak-drive}

The strong-drive limits above are distinct from the regular weak-drive
problem.  Put
\begin{equation}
 \delta=\frac{\hbar\Delta}{V},\qquad f=\frac{F}{V},\qquad
 \lambda=|f|\ll1.
\end{equation}
After a number-operator phase rotation we may take $f=\lambda>0$, and
\begin{equation}
 h=\frac{H}{V}=h_0+\lambda(z+\partial_z),\qquad
 h_0=\frac{1}{2}z^2\partial_z^2+\delta z\partial_z .
 \label{eq:weak-drive-split}
\end{equation}
In particular, the entire number term belongs to the zeroth-order operator.
Its normalized eigenfunction and energy are
\begin{equation}
 \Psi_N^{(0)}(z)=\frac{z^N}{\sqrt{N!}},\qquad
 e_N=\frac{N(N-1)}{2}+\delta N,\qquad N\in\Nzero .
 \label{eq:weak-drive-zero}
\end{equation}
For fixed $N$, nondegenerate perturbation theory requires
$\delta\ne-(n+N-1)/2$ for every $n\in\Nzero\setminus\{N\}$.

With intermediate normalization, parity makes the energy an even function
of $\lambda$.  Direct Rayleigh--Schr\"odinger calculation, independently
checked in the Bargmann coefficient recurrence, gives
\begin{equation}
 \frac{E_N}{V}=e_N+\lambda^2 e_N^{(2)}+\lambda^4e_N^{(4)}+O(\lambda^6),
 \qquad
 e_N^{(2)}=\frac{N}{N-1+\delta}-\frac{N+1}{N+\delta},
 \label{eq:weak-drive-energy2}
\end{equation}
and
\begin{align}
e_N^{(4)}={}&\frac{N(N-1)}{(N-1+\delta)^2(2N-3+2\delta)}
-\frac{(N+1)(N+2)}{(N+\delta)^2(2N+1+2\delta)}\notag\\
&-e_N^{(2)}\left[\frac{N}{(N-1+\delta)^2}
+\frac{N+1}{(N+\delta)^2}\right].
\label{eq:weak-drive-energy4}
\end{align}
The first term in \eqref{eq:weak-drive-energy4} is absent for $N<2$; for
$N=0$, every lower-channel contribution proportional to $N$ is omitted
before division.  Thus $e_0^{(2)}=-1/\delta$ and
$e_0^{(4)}=1/[\delta^3(1+2\delta)]$, giving $e_0^{(4)}=1/3$ at $\delta=1$.
If
$\Delta_j=e_N-e_{N+j}$, the wavefunction begins
\begin{align}
|\Psi_N\rangle={}&|N\rangle+\lambda\left(
 \frac{\sqrt N}{\Delta_{-1}}|N-1\rangle
 +\frac{\sqrt{N+1}}{\Delta_1}|N+1\rangle\right)\notag\\
&+\lambda^2\left(
 \frac{\sqrt{N(N-1)}}{\Delta_{-1}\Delta_{-2}}|N-2\rangle
 +\frac{\sqrt{(N+1)(N+2)}}{\Delta_1\Delta_2}|N+2\rangle\right)
 +O(\lambda^3).
\label{eq:weak-drive-wavefunction}
\end{align}
At order $k$, only $|n-N|\le k$ and $n-N\equiv k\pmod2$ occur.

Bender and Bettencourt reorganized growing polynomial diagonals of the
anharmonic-oscillator perturbation series and iterated that operation until
its large-order exponent matched WKB \cite{BenderBettencourt1996Multiple}.
For the present tridiagonal problem, the directly analogous monotone upper
edge can be summed exactly:
\begin{equation}
 \frac{\Psi_N(z;\lambda)}{z^N/\sqrt{N!}}\Big|_{\rm upper\ edge}
 ={}_0F_1(;2N+2\delta;-2\lambda z).
 \label{eq:weak-drive-0f1}
\end{equation}
It reorganizes the scaling $\lambda z=O(1)$.  The lower edge is only the
terminating polynomial
${}_1F_1(-N;-2N+2-2\delta;2\lambda/z)$; its apparent poles cancel after
multiplication by $z^N$.  These are selected edge families, not a sum of
the complete perturbation series.  If $b=2N+2\delta$ and
$x=-2\lambda z$, then
\begin{equation}
 {}_0F_1(;b;x)=\Gamma(b)x^{(1-b)/2}I_{b-1}(2\sqrt{x})
 =\Gamma(b)(-x)^{(1-b)/2}J_{b-1}(2\sqrt{-x}),
 \label{eq:weak-drive-bessel}
\end{equation}
with consistent branches and $b\notin\{0,-1,-2,\ldots\}$.  Its two formal
large-argument exponentials are
$\exp(\pm2\sqrt{-2\lambda z})$.  Independently, the exact equation is put
in normal form by $\Psi=z^{-\delta}e^{\lambda/z}u$, where
\begin{equation}
 u''+\left\{\frac{2\lambda}{z}
 +\frac{-2\epsilon+\delta-\delta^2}{z^2}
 +\frac{2(1-\delta)\lambda}{z^3}-\frac{\lambda^2}{z^4}\right\}u=0.
 \label{eq:weak-drive-normal-form}
\end{equation}
Hence the formal leading WKB exponent and transport power are respectively
$\pm2\sqrt{-2\lambda z}$ and $z^{1/4-\delta}$.  The dominant Bessel
asymptotic has the same power because
$z^Nx^{(1-2b)/4}\propto z^{1/4-\delta}$.  This comparison requires the
further overlap $1\ll|x|$ in a sector with a fixed square-root branch; it
is not a global WKB theorem.

The first upper subdiagonal can be incorporated by a genuine multiplicative
step.  Write $Y_0={}_0F_1(;b;x)$, $r=Y_0'/Y_0$, and
\begin{equation}
 \frac{\Psi_N}{z^N/\sqrt{N!}}=Y_0(x)Z(x),\qquad
 Z=1+\lambda^2Z_1+O(\lambda^4).
\end{equation}
Direct substitution in the exact Bargmann equation gives
\begin{equation}
 \mathcal A_bZ-\lambda^2\left[2Z'+2\left(r+\frac Nx\right)Z\right]
 -(\epsilon-e_N)Z=0,
 \quad
 \mathcal A_b=\frac{x^2}{2}\partial_x^2+
 \left(x^2r+\frac{bx}{2}\right)\partial_x .
 \label{eq:weak-drive-dressed}
\end{equation}
Hence
$\mathcal A_bZ_1=2(r+N/x)+e_N^{(2)}$ and locally
$Z_1=Y_1/Y_0$, where $Y_1$ is the additive sum.
For $N\ge1$, $Y_1=-4N[(b-2)x]^{-1}+O(x)$; for $N=0$ the
$x^{-1}$ term is absent, including the regular case $b=2$.
Zeros of $Y_0$ generically become poles of $Z_1$, although
$Y_0Z_1=Y_1$ remains regular, so this factorization is local.

Writing $Z=\exp(\lambda^2S_1+\lambda^4S_2+\cdots)$ gives
$\mathcal A_bS_1=2(r+N/x)+e_N^{(2)}$ and
\begin{equation}
 \mathcal A_bS_2=2S_1'+e_N^{(4)}-\frac{x^2}{2}(S_1')^2.
 \label{eq:weak-drive-second-dressing}
\end{equation}
The Laurent solvability condition reproduces $e_N^{(4)}$.  In a dominant
Bessel sector,
$S_1=-2e_N^{(2)}x^{-1/2}-(2+e_N^{(2)}/2)x^{-1}+O(x^{-3/2})$,
matching the next energy, gauge, and transport terms of formal WKB.
Summing $Y_2$ by a normalized Green integral verifies
$S_2=Y_2/Y_0-S_1^2/2$.  The latter has generic double poles at
zeros of $Y_0$, whereas the reconstructed additive state is entire and
has rapidly convergent Fock coefficients.  The perturbative expansion of
its Rayleigh quotient agrees with the exact Rayleigh--Schr\"odinger series
through orders $\lambda^{12},\lambda^8,\lambda^6$, and $\lambda^4$ for
$N=0,1,2$, and $N\ge3$, respectively; the level dependence reflects
omitted lower diagonals.

This section develops a regular Rayleigh--Schr\"odinger perturbation
expansion of the weak-drive eigenproblem in the parameter $\lambda=|F|/V$,
obtained after a number-operator phase rotation, giving the energy
corrections through order $\lambda^4$
\eqref{eq:weak-drive-energy2}--\eqref{eq:weak-drive-energy4}, independently
checked against the Bargmann coefficient recurrence.  It then shows that
the monotone upper-edge diagonal of the resulting tridiagonal perturbation
series sums exactly to a confluent hypergeometric closed form
\eqref{eq:weak-drive-0f1}, in analogy with the diagonal resummation of
Bender and Bettencourt, and extends this resummation multiplicatively to
the first subdiagonal, matching the exact perturbative energy series
through the orders recorded above.

In principle the full multiplicative series could define a global Bargmann
state only after its moving zeros and growth are controlled.  A finite
exponential truncation cannot: cumulant poles at zeros of $Y_0$ become
essential singularities, while a pure exponential is zero-free.  A global
form would likely separate a zero-carrying canonical product from its
exponential factor.  Direct multiple scales or quasilinearization of the
logarithmic-derivative Riccati equation are plausible local or sectorial
directions, but moving poles, Stokes data, and Bargmann admissibility leave
their practical closure open.  This is not a closed global hierarchy, and
the Bessel rather than Weber description remains sectorial and zero-dependent.

\section{Summary}
\label{sec:summary}

This section collects, for reference, every substantive mathematical claim
made above, together with its precise epistemic status: \emph{proven}
(established unconditionally, within the stated standing hypotheses, such
as $F\ne0$ where noted), \emph{conditional} (following from an explicitly
named, unproved hypothesis or conjecture stated elsewhere in the paper), or
\emph{conjectured} (itself an unproved hypothesis proposed here). No new
mathematical content is introduced in this section; each item below
restates, with a cross-reference, a claim already made in the body of the
paper, and does not go beyond its strongest formulation there.

\begin{enumerate}
\item \textbf{Self-adjointness and spectral discreteness (proven).} The
 Hamiltonian~\eqref{eq:H} is self-adjoint on $\Dom(N^2)$, essentially
 self-adjoint on the finite-particle core, bounded below, and has compact
 resolvent, so its spectrum is purely discrete with no finite accumulation
 point (Proposition~\ref{prop:self-adjoint}, Proposition~\ref{prop:lower-bound},
 and Theorem~\ref{thm:compact-resolvent}); for the single-mode case this specializes a
 more general driven-lattice result, with explicit constants supplied here
 directly in the Bargmann representation.
\item \textbf{Continued-fraction spectral characterization (proven, for
 $F\ne0$).} By the asymptotic dichotomy of
 Lemma~\ref{lem:asymptotic-dichotomy} and
 Theorem~\ref{thm:bargmann-characteristic}, a complex number $E$ is an
 eigenvalue of $H$ if and only if it satisfies the exact scalar condition
 $\Xi(E)=0$ \eqref{eq:characteristic-cf}, equivalently if and only if the
 physical-slice recurrence \eqref{eq:taylor-recurrence} selects its minimal
 solution, equivalently if and only if the resulting Bargmann series
 defines an entire function of order $1/2$; every such eigenvalue is
 simple.
\item \textbf{The undriven spectrum (proven, exact).} For $F=0$ the
 spectrum is exactly $\{E_n^{(0)}\}_{n\in\Nzero}$ of
 \eqref{eq:undriven-spectrum}, with eigenfunctions $z^n/\sqrt{n!}$; the same
 degeneration reappears, consistently, at the level of the Liouville normal
 form as the reduction of \eqref{eq:standard-ode} to an Euler equation in
 Section~\ref{sec:liouville}.
\item \textbf{Singularity structure and formal local behaviours (proven,
 as formal statements).} The first-derivative-free normal form
 \eqref{eq:liouville-normal-form} has an unramified irregular singularity
 of Poincar\'e rank $1$ at the origin and a ramified irregular singularity
 of rank $1/2$ at infinity (Proposition~\ref{prop:singularity-classification}),
 with two explicit leading formal behaviours at each point
 (\eqref{eq:phi-regular-formal}--\eqref{eq:phi-singular-formal} at the
 origin, \eqref{eq:infinity-formal-behaviours} at infinity) and an explicit
 compensation mechanism
 (\eqref{eq:compensated-branch}--\eqref{eq:enhanced-branch}) distinguishing
 them; sectorial asymptotic existence and Stokes connection data are not
 asserted by these formal statements alone.
\item \textbf{Identification with the double-confluent Heun family (proven,
 algebraic correspondence only).} The normal-form equation is exactly the
 degenerate ($\varepsilon_{\rm D}=0$) parameter specialization
 \eqref{eq:dche-correspondence} of the canonical double-confluent Heun
 family, and the Whittaker--Ince equation separately arises from a distinct
 limiting operation \eqref{eq:whittaker-ince-limit};
 neither correspondence, by itself, identifies a preferred global Heun
 solution, proves a connection formula, or yields an energy-quantization
 condition.
\item \textbf{Normalized sectorial solutions at infinity (proven, for
 $F\ne0$).} Theorem~\ref{thm:olver-sectorial} constructs, by Olver's
 fixed-parameter error-control method, unique holomorphic canonical
 solutions on suitable half-planes near infinity, with explicit remainder
 bounds \eqref{eq:olver-epsilon-bound}--\eqref{eq:olver-radial-majorant} and
 no large-parameter assumption.
\item \textbf{Strong-drive energy, leading order (proven), next order
 (conditional).} In the strong-drive limit, the leading coalescing-turning-point
 energy coefficient $-\tfrac32\eta^{4/3}$ in \eqref{eq:strong-drive-energy}
 is proved rigorously from the double-root turning-point geometry alone,
 while its next coefficient in the same display is conditional on
 completing the global connection step formalized by
 Conjecture~\ref{conj:essential-connection}.
\item \textbf{Global connection coefficient (conjectured).}
 Conjecture~\ref{conj:essential-connection} proposes that a global
 compensated-to-essential connection coefficient $C_{\rm ess}$, whose
 vanishing would give the exact strong-drive quantization condition, takes
 the displayed near-Weber form; the underlying domain chain, symmetry
 reduction, determinant identification, and remainder control are all
 conjectural.
\item \textbf{Independent Bogoliubov energy expansion (proven, as a formal
 algebraic identity; not a statement about the true spectrum beyond item~7's
 rigorous leading term).} Independently of
 Conjecture~\ref{conj:essential-connection}, a Bogoliubov operator expansion
 yields the algebraically verified formal Rayleigh--Schr\"odinger energy
 series \eqref{eq:strong-drive-energy-constant} through order
 $\eta^{-2/3}$; its agreement with the leading two terms of item~7 is a
 consistency check, not a proof, of that conjecture.
\item \textbf{Complex-WKB branches and the Fedoryuk--Bargmann connection
 (conjectured/conditional).} An independent complex-WKB treatment directly
 in the Fedoryuk normal form identifies two candidate sectorial energy
 branches \eqref{eq:fedoryuk-local-energies}; under the (conjectured)
 Fedoryuk--Bargmann orientation hypothesis of
 Section~\ref{sec:sectorial-solutions} selecting the imaginary-axis branch,
 and under the further Global Fedoryuk--Bargmann connection hypotheses
 H1--H3 of Conjecture~\ref{conj:fedoryuk-bargmann-lemma}, the resulting
 conditional quantization condition \eqref{eq:fedoryuk-conditional-quantization}
 agrees with item~9's formal series; both the orientation choice and H1--H3
 are unproved, and this agreement is again only a consistency check.
\item \textbf{Weak-drive perturbative energies (proven, as a formal
 perturbative expansion).} Direct Rayleigh--Schr\"odinger perturbation
 theory in the weak-drive parameter $\lambda=|F|/V$ gives the energy
 corrections \eqref{eq:weak-drive-energy2}--\eqref{eq:weak-drive-energy4}
 through order $\lambda^4$, independently reproduced from the Bargmann
 coefficient recurrence.
\item \textbf{Exact upper-edge resummation (proven, exact closed form).}
 The monotone upper-edge diagonal of the resulting tridiagonal perturbation
 series sums exactly to the confluent hypergeometric closed form
 \eqref{eq:weak-drive-0f1}, in direct analogy with the diagonal resummation
 of Bender and Bettencourt.
\end{enumerate}

\section{Outlook}
\label{sec:outlook}

Taken as a whole, this work gives a self-adjoint, discrete-spectrum
formulation of the driven Kerr eigenproblem in the Bargmann representation
(Section~\ref{sec:operator-properties}), an exact scalar spectral condition
$\Xi(E)=0$ requiring no matrix diagonalization
(Section~\ref{sec:heunwi}), a complete classification of the equation's
singularity structure and its place in the double-confluent Heun family
(Section~\ref{sec:liouville}), rigorously controlled sectorial solutions at
infinity (Section~\ref{sec:sectorial-infinity}), and two independently
cross-checked asymptotic pictures of the strong- and weak-drive regimes
(Sections~\ref{sec:strong-drive}--\ref{sec:weak-drive}).  The one remaining
open step, a global origin-to-infinity connection that would turn the
strong-drive WKB pictures into a proof, is stated only as a conjecture here
(Summary, items 7--10).  Beyond that step, comparison of the present
degenerate ($\varepsilon_{\rm D}=0$) parameter case with the broader
double-confluent Heun literature
\cite{Buhring1994DCHE,ElJaickFigueiredo2008Solutions,
ElJaickFigueiredo2013Coulomb,IshkhanyanEtAl2014DCHE,Leaver1986Spectral}, and
explicit comparison of the resulting holomorphic eigenfunctions with
existing coherent-state, Husimi, and Wigner descriptions
\cite{MaslovaEtAl2019Squeezing,KirchmairEtAl2013Kerr}, are natural
directions for future work.

\appendix
\section{Exploratory WKB iteration portraits}
\label{app:wkb-portraits}

For complementary exploratory visualization of the local WKB dynamics,
we present iteration portraits derived from the sectorial Fedoryuk--Weber
analysis of Section~\ref{sec:sectorial-solutions}. While this is not physical
time evolution and does not prove global Bargmann admissibility, the
resulting regularized iterated map provides qualitative insight into the
connection structure.

For exploratory visualization only, we treat the local, sectorial
Fedoryuk--Weber expression \eqref{eq:fedoryuk-local-pullback} formally as a
self-map of the complex plane.  Starting from each plotted pixel $z_0$, we
compute
\begin{equation}
 z_{j+1}=\mathcal R_{\rm post}\!\left[
  \Psi_{n,\sigma}^{\rm loc}(z_j;\mathcal R_{\rm eval})\right],
 \label{eq:regularized-local-iteration}
\end{equation}
where \(\mathcal R_{\rm eval}\) acts during evaluation and
\(\mathcal R_{\rm post}\) afterwards.  This is not physical time evolution,
and \(\Psi^{\rm loc}\) is not a proved global Bargmann eigenfunction.

The numerical prescription is deliberately explicit.  A complex denominator
\(d\) with \(|d|<\epsilon_{\rm den}=10^{-8}\) is replaced by
\(10^{-8}e^{i\arg d}\), with phase zero at \(d=0\).  Only the real part of
each complex exponent is clipped to \([-80,80]\), leaving its imaginary part
unchanged.  After every application, newly nonfinite values are marked invalid
and that mask is propagated; every remaining value with modulus above
\(z_{\rm cap}=10^6\) is capped radially, preserving phase.  Invalid pixels
are black.  NumPy principal branches are used for logarithms and powers;
branch cuts are neither tracked nor analytically continued, so their
discontinuities remain visible.  The regularized map, not the unmodified
analytic expression alone, generates the pixels.  No smoothing, interpolation
over singularities, denoising, or post-hoc image enhancement is applied, and
no claim or test of fractality is made.

In all panels hue is \((\arg z_k+\pi)/(2\pi)\), saturation is \(0.97\), and
valid brightness lies in \([0.50,0.98]\).  With
\(t\) the clipped normalization of \(\log(1+|z_k|)\) between its finite
2nd and 98th percentiles, phase--modulus uses
\(0.50+0.48t^{0.65}\); banded modulus replaces \(t^{0.65}\) by
\(0.15t+0.85[1+\sin(2\pi\,2.4t)]/2\); equalized brightness uses the
deterministic empirical midrank CDF.  Thus the colour rules expose different
features rather than merely recolouring identical data.

The intervention rates (percent of all pixels) are:
\begin{center}
\small
\begin{tabular}{cclrrr}
\hline
\(\eta\)&branch&application&denominator floor&exponent clip&radial cap\\ \hline
4&R&1&0&0.0169&0.2202\\
 &&2&0.1762&22.8771&23.6858\\
4&I&1&0&0&0\\
 &&2&0&36.2521&7.9441\\
0.25&R&1&0&0.0003&0.0041\\
 &&2&0.0013&0.1881&0.1988\\
 &&3&0.1918&4.1502&4.3292\\
0.25&I&1&0&0&0\\
 &&2&0&0&0\\
 &&3&0&0&0.5218\\ \hline
\end{tabular}
\end{center}
New-nonfinite and total-invalid rates are zero in every row.  The large
second-application rates at \(\eta=4\) show how strongly the disclosed
numerical hammer acts.  Conversely, the milder rates at \(\eta=0.25\) do not
validate the approximation: those panels are formal extrapolations far
outside the controlled strong-drive asymptotic regime.

\begin{figure}[p]
 \centering
 \begin{minipage}[t]{.49\linewidth}\centering
  \includegraphics[width=\linewidth]{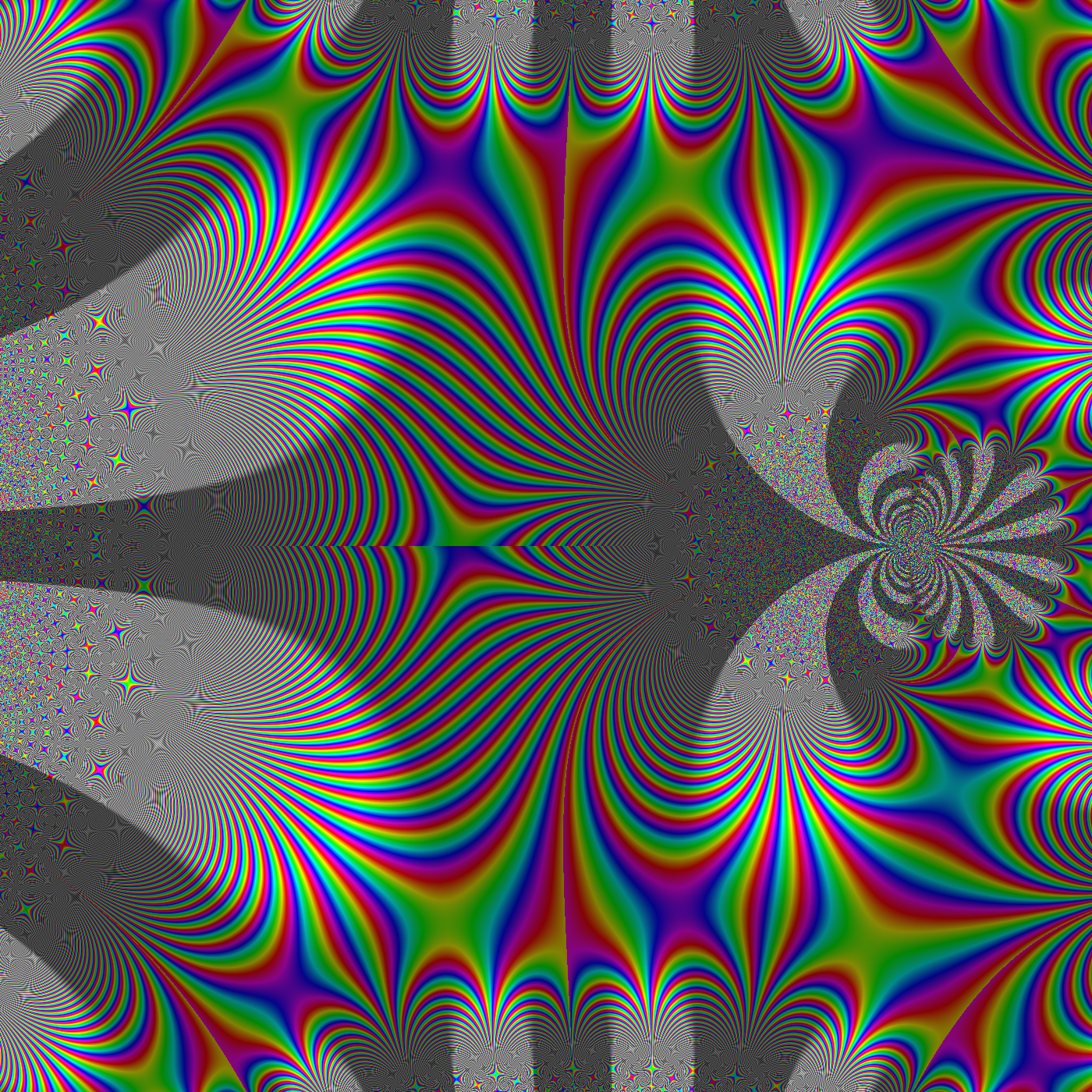}\\[-1mm]
  \textbf{(a)}
 \end{minipage}\hfill
 \begin{minipage}[t]{.49\linewidth}\centering
  \includegraphics[width=\linewidth]{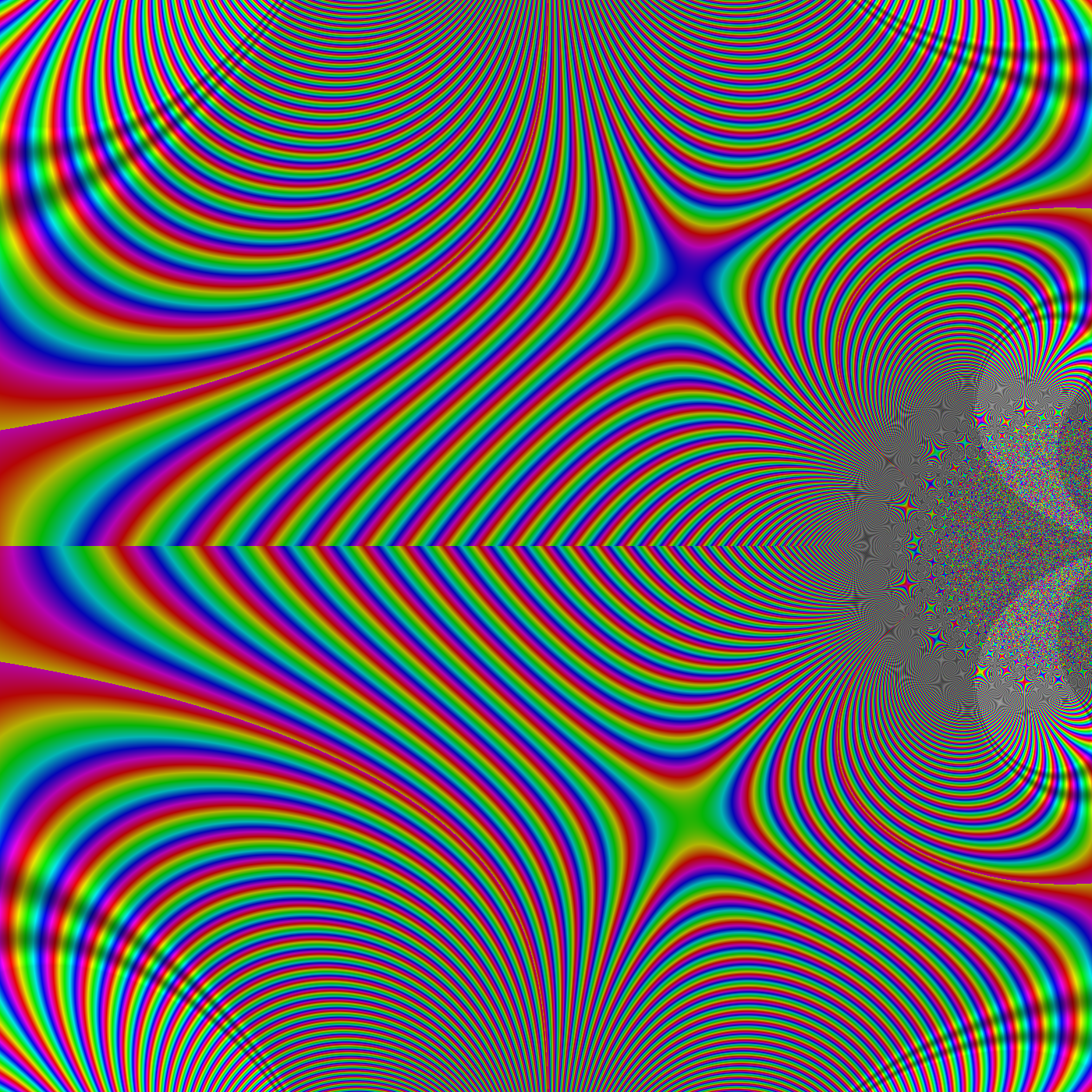}\\[-1mm]
  \textbf{(b)}
 \end{minipage}
 \begin{minipage}[t]{.49\linewidth}\centering
  \includegraphics[width=\linewidth]{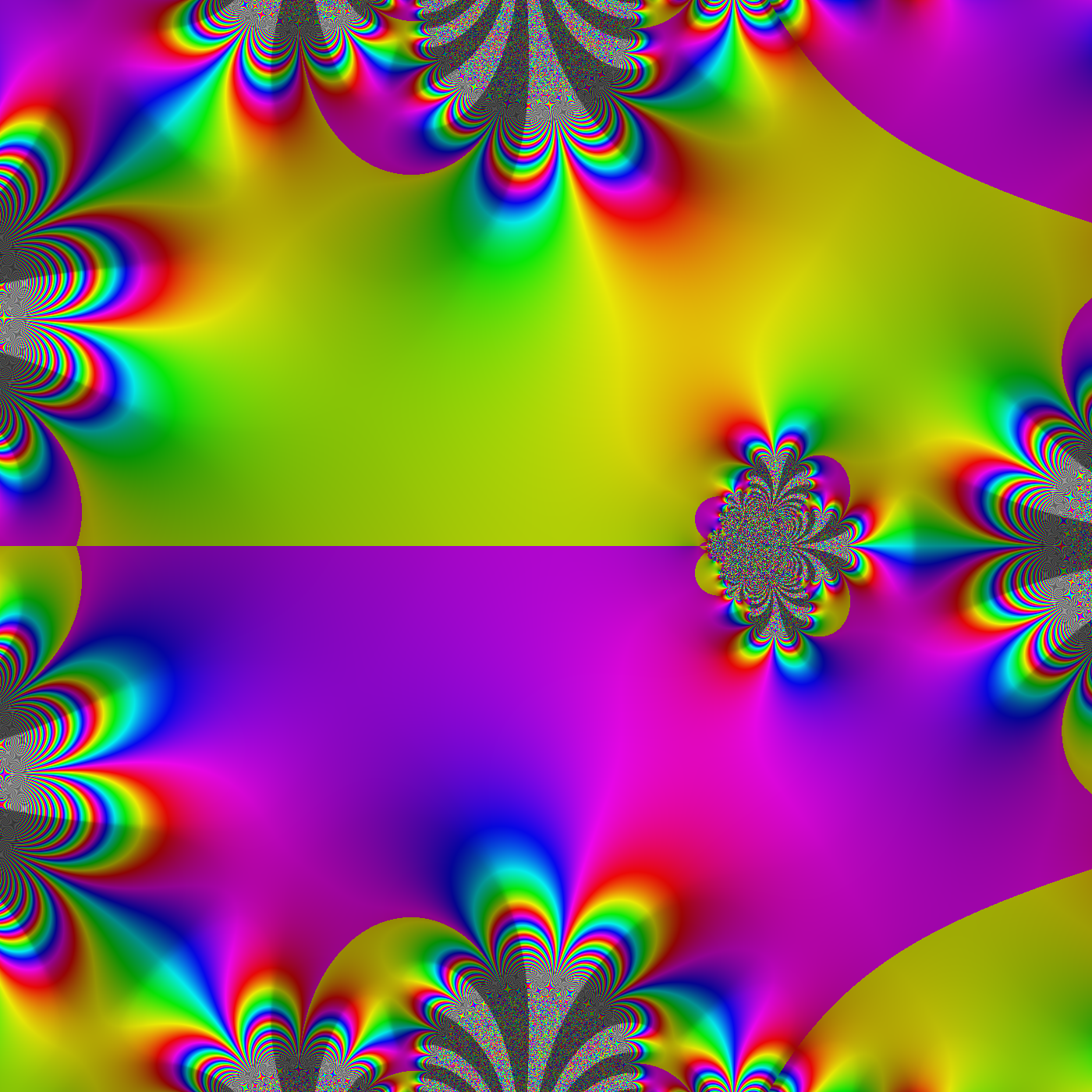}\\[-1mm]
  \textbf{(c)}
 \end{minipage}\hfill
 \begin{minipage}[t]{.49\linewidth}\centering
  \includegraphics[width=\linewidth]{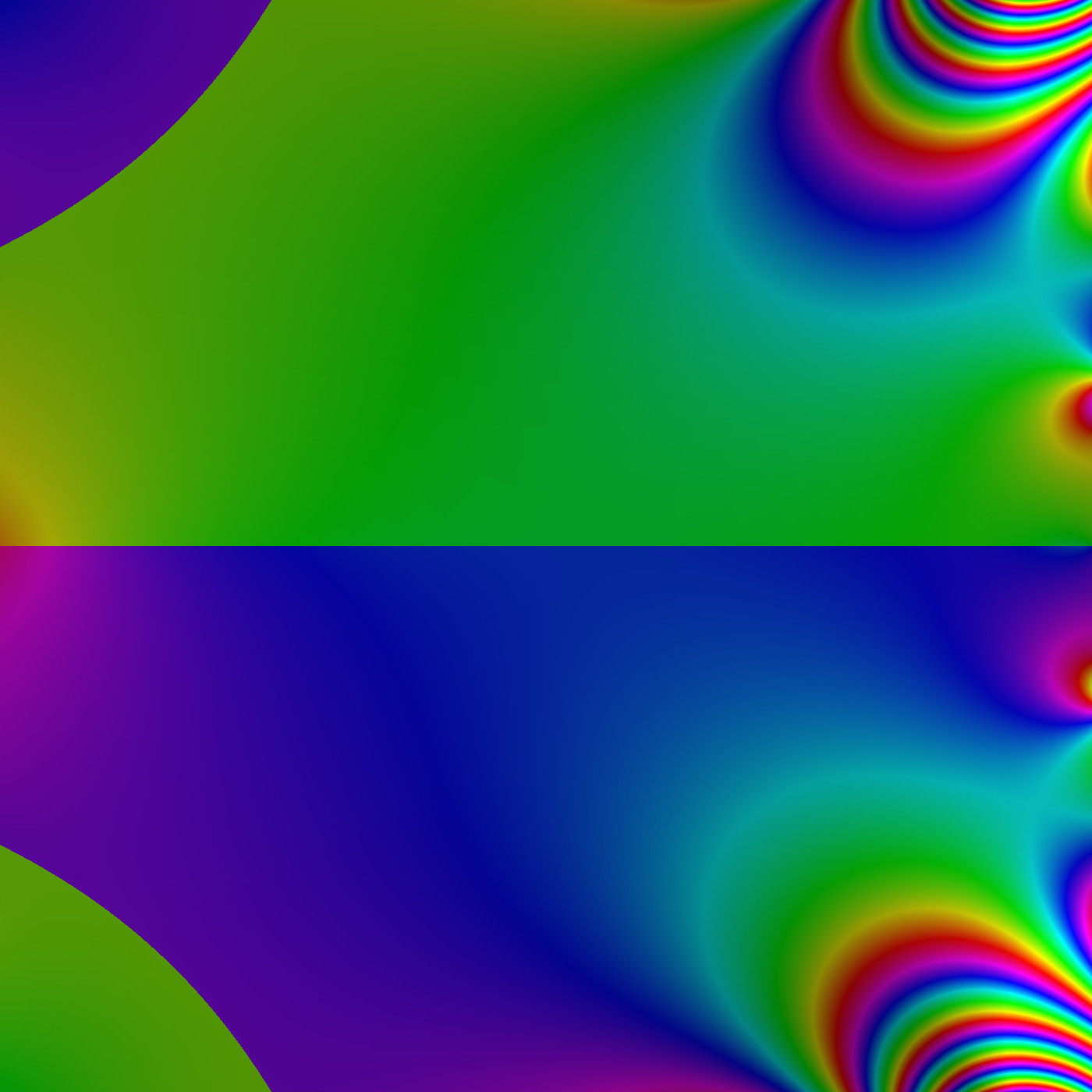}\\[-1mm]
  \textbf{(d)}
 \end{minipage}
 \caption{Exploratory portraits with \(\delta=0.7\), \(n=0\), principal
 branches, and the naive regularization defined in the text.  (a)
 \(\eta=4\), \(k=2\), branch \(R\), domain
 \([-4,0.8]\times[-2.4,2.4]\), phase--modulus colour.  (b) \(\eta=4\),
 \(k=2\), branch \(I\), domain \([-3.2,-0.2]\times[-1.5,1.5]\), smooth
 banded-modulus colour.  (c) \(\eta=0.25\), \(k=3\), branch \(R\), domain
 \([-2,1]\times[-1.5,1.5]\), equalized phase/modulus colour; this is a
 formal extrapolation far outside the controlled strong-drive asymptotic
 regime.  (d) \(\eta=0.25\), \(k=3\), branch \(I\), domain
 \([-1.8,-0.1]\times[-1,1]\), phase--modulus colour, with the same explicit
 out-of-regime warning.  Hue encodes phase and brightness encodes modulus in
 (a), (b), and (d), and the empirical modulus rank in (c).}
 \label{fig:iterated-wkb-final}
\end{figure}

Visually, the strong-drive pair contains dense nested phase fans: the wider
\(R\) window emphasizes repeated lobes, while the local \(I\) window emphasizes
a symmetric central seam and smooth modulus bands.  The extrapolative \(R\)
panel has broad colour fields punctuated by compact rosettes, whereas the
extrapolative \(I\) panel is quieter except near its right boundary.  These
descriptions characterize the regularized iterated map defined above.  They
are exploratory visualizations, with no interpretation as physical evolution,
a global Bargmann eigenfunction, or evidence of fractality.

\clearpage
\section*{Acknowledgment}

During the preparation of this work, generative AI tools were used
in two distinct capacities. First, AI-assisted code generation
(ChatGPT Work, Claude Opus, Claude Code, and Codex) was used for creation
of Python routines; all final implementation was inspected and verified
by the author. Second, large language models (ChatGPT Work, Claude Opus)
were used to assist with prose editing and language refinement at the
manuscript stage. The author reviewed and took responsibility for all
content thus produced. No AI tool is listed as an author; the author
accepts full responsibility for the scientific content of this work,
consistent with both arXiv and the Journal of Mathematical Physics
policies.

\bibliographystyle{unsrt}
\bibliography{manuscript}

\end{document}